\documentclass[pra,
superscriptaddress,
twocolumn,
]{revtex4-1}
\usepackage[latin1]{inputenc}
\usepackage{graphicx}
\usepackage{amsmath}
\usepackage{amsthm}
\usepackage{bm}
\usepackage{layout}
\usepackage{float}
\usepackage{amsfonts}
\usepackage{amssymb}%
\usepackage[margin=0.75in]{geometry}
\usepackage{color}
\usepackage{soul}
\usepackage{mathtools}
\usepackage[colorlinks=true,citecolor=blue]{hyperref}

\newtheorem{lemma}{Lemma}
\newtheorem{definition}{Definition}

\newtheorem{corollary}{Corollary}
\usepackage[capitalise,compress]{cleveref}
\newcommand{\avg}[1]{\left \langle #1 \right\rangle}

\newcommand{\ket}[1]{\left | #1 \right\rangle}
\newcommand{\bra}[1]{\left \langle #1 \right |}

\newcommand{\Tr}{\mathrm{Tr}}

\newcommand{\abs}[1]{\left | #1 \right|}
\renewcommand{\epsilon}{\varepsilon}
\renewcommand{\O}[1]{\mathcal O\left(#1\right)}
\newcommand{\norm}[1]{\left\|#1\right\|}
\newcommand{\lnorm}[1]{\left\|#1\right\|_l}
\newcommand{\comm}[1]{\left[#1\right]}

\newcounter{para}
\newcommand{\dist}{\textrm{dist}}

\makeatletter
\newcommand*\bigcdot{\mathpalette\bigcdot@{.5}}
\newcommand*\bigcdot@[2]{\mathbin{\vcenter{\hbox{\scalebox{#2}{$\m@th#1\bullet$}}}}}

\makeatother
\newcommand{\nmax}{\omega_*}
\newcommand{\qmax}{q_{\text{max}}}
\newcommand{\ad}{\text{ad}}
\newcommand{\Hpl}{\mathbb{H}_{\alpha}}
\newcommand{\dash}{\textemdash}

\usepackage{braket}

\usepackage{array}
\newcolumntype{L}{>{$}l<{$}} % math-mode version of "l" column type
\newcolumntype{C}{>{$}c<{$}} % math-mode version of "c" column type
\newcolumntype{R}{>{$}r<{$}} % math-mode version of "r" column type

\usepackage{xcolor}

\usepackage{mathtools}
\DeclarePairedDelimiter\ceil{\lceil}{\rceil}
\DeclarePairedDelimiter\floor{\lfloor}{\rfloor}

\usepackage{xr}
\makeatletter
\newcommand*{\addFileDependency}[1]{% argument=file name and extension
  \typeout{(#1)}
  \@addtofilelist{#1}
  \IfFileExists{#1}{}{\typeout{No file #1.}}
}
\makeatother

\newcommand{\etal}{{\it et al.}}
\begin{document}

\title{\vspace*{0in} Locality and Heating in Periodically Driven, Power-law Interacting Systems}
\author{Minh C.\ Tran}
\affiliation{Joint Center for Quantum Information and Computer Science, NIST/University of Maryland, College Park, MD 20742, USA}
\affiliation{Joint Quantum Institute, NIST/University of Maryland, College Park, MD 20742, USA}
\affiliation{Kavli Institute for Theoretical Physics, University of California, Santa Barbara, CA 93106, USA}
\author{Adam Ehrenberg}
\affiliation{Joint Center for Quantum Information and Computer Science, NIST/University of Maryland, College Park, MD 20742, USA}
\affiliation{Joint Quantum Institute, NIST/University of Maryland, College Park, MD 20742, USA}
\author{Andrew Y. Guo}
\affiliation{Joint Center for Quantum Information and Computer Science, NIST/University of Maryland, College Park, MD 20742, USA}
\affiliation{Joint Quantum Institute, NIST/University of Maryland, College Park, MD 20742, USA}
\author{Paraj Titum}
\affiliation{Joint Center for Quantum Information and Computer Science, NIST/University of Maryland, College Park, MD 20742, USA}
\affiliation{Joint Quantum Institute, NIST/University of Maryland, College Park, MD 20742, USA}
\affiliation{Johns Hopkins University Applied Physics Laboratory, Laurel, Maryland 20723, USA}
\author{Dmitry A. Abanin}
\affiliation{Department of Theoretical Physics, University of Geneva, 1211 Geneva, Switzerland}
\author{Alexey V.\ Gorshkov}
\affiliation{Joint Center for Quantum Information and Computer Science, NIST/University of Maryland, College Park, MD 20742, USA}
\affiliation{Joint Quantum Institute, NIST/University of Maryland, College Park, MD 20742, USA}
\affiliation{Kavli Institute for Theoretical Physics, University of California, Santa Barbara, CA 93106, USA}
\date{\today}
\begin{abstract}
We study the heating time in periodically driven $D$-dimensional systems with interactions that decay with the distance $r$ as a power-law $1/r^\alpha$.
Using linear response theory, we show that the heating time is exponentially long as a function of the drive frequency for $\alpha>D$.
For systems that may not obey linear response theory, we use a more general Magnus-like expansion to show the existence of quasi-conserved observables, which imply exponentially long heating time, for $\alpha>2D$.
We also generalize a number of recent state-of-the-art Lieb-Robinson bounds for power-law systems from two-body interactions to $k$-body interactions and thereby obtain a longer heating time than previously established in the literature. 
Additionally, we conjecture that the gap between the results from the linear response theory and the Magnus-like expansion does not have physical implications, but is, rather, due to the lack of tight Lieb-Robinson bounds for power-law interactions.
We show that the gap vanishes in the presence of a hypothetical, tight bound. 

\end{abstract}\maketitle
\newpage
\newpage
\section{Introduction}
Periodically driven systems can host interesting non-equilibrium physics, such as Floquet topological insulators~\cite{Cayssol2013}, time crystals~\cite{Moessner2017,Else2019TimeCrystal}, and anomalous Floquet phases~\cite{Harper2019}. 
However, most driven systems eventually heat up to equilibrium, infinite-temperature states, erasing the interesting features in the process.

The timescale before heating becomes appreciable in periodically driven systems is known as the \emph{heating time}, and it generally exhibits a nontrivial dependence on the frequency of the drive, $\omega$.
Previous works~\cite{Abanin15,Abanin17PRB,Abanin17CMP,Kuwahara2016} established that finite-range interacting systems under rapid, local~\footnote{The drive is local if it can be written as a sum of local terms.}, periodic drives could not heat up until after a time $t_*=e^{\O{\omega}}$ that is exponentially long in the drive frequency $\omega$.
This slow \emph{heating rate} stems at least in part from the locality of the interactions, which constrains the probability that distant particles collectively absorb an energy quantum $\hbar \omega$.

This result also applies to systems with long-range interactions that decay with the distance $r$, e.g.~as a power-law $1/r^\alpha$. 
Such systems are of great interest as they can be implemented in a wide variety of experiments, such as trapped ions~\cite{Britton2012,Kim2011}, Rydberg atoms~\cite{Saffman10}, ultracold atoms and molecules~\cite{Douglas2015,Yan2013}, nitrogen-vacancy centers~\cite{Maze2011}, and superconducting circuits~\cite{Otten16}. 
On the theoretical side, for spin systems with disordered, sign-changing power-law couplings, Ref.~\cite{Abanin18} demonstrated the exponentially-suppressed heating rate when $\alpha>D/2$, where $D$ is the dimensionality of the system.
Furthermore, Ref.~\cite{Kuwahara2016} proved an exponential heating time $t_*=e^{\O{\omega}}$ for general power-law interactions with $\alpha>2D$.  
In contrast, for $D<\alpha<2D$, Ref.~\cite{Kuwahara2016} only obtained a linear heating time $t_*=\O{\omega}$, while numerical evidence~\cite{Machado2017} suggests that the heating time is still exponential within this regime of $\alpha$.

In this paper, we study the heating time in periodically driven, power-law interacting systems with $\alpha>D$ from two different perspectives.
Within linear response theory, we show that such systems only heat up after some time exponentially large as a function of the drive frequency.
This result mirrors the statement established for finite-range interactions in Ref.~\cite{Abanin15} and extends Ref.~\cite{Abanin18} to systems without disorder (though at the expense of a smaller range of valid $\alpha$).
The result also matches the numerical evidence in Ref.~\cite{Machado2017}. 
For generic periodically driven, power-law interacting systems that may not obey the linear response theory\dash such as those under a strong drive\dash we generalize Ref.~\cite{Abanin17PRB} and construct an effective time-independent Hamiltonian with power-law interactions.
This Hamiltonian closely describes the dynamics of the driven system up to time $t_*$, where $t_*$ is exponentially large as a function of the drive frequency.
We thereby show that the system cannot heat up until at least after this timescale.

We note that, although our generalization of Ref.~\cite{Abanin17PRB} is different from Ref.~\cite{Kuwahara2016}, it is similar in spirit to their arguments. 
While Ref.~\cite{Kuwahara2016} mainly focused on finite-range interactions, their construction of the effective Hamiltonian by truncating the Magnus series would also apply to power-law systems.
However, our approach here also provides insights into the structure of the effective Hamiltonian. 
In particular, we show that the effective Hamiltonian is also power-law with the same exponent $\alpha$ as the driven Hamiltonian.
Furthermore, we prove a stronger, albeit still exponential in $\omega$, bound on the heating time than one would get from the argument in Ref.~\cite{Kuwahara2016}. 
This improvement relies on the use of state-of-the-art Lieb-Robinson bounds~\cite{Else18,Tran18}, which we develop for this purpose. 
In particular, through a new technique, we generalize the bound in Ref.~\cite{Tran18} from two-body to many-body interactions.

Similarly to Ref.~\cite{Kuwahara2016}, our construction requires $\alpha>2D$, in contrast to the numerical evidence in Ref.~\cite{Machado2017} and to the wider range of validity $\alpha>D$ found in the linear response theory.
Because both Ref.~\cite{Kuwahara2016} and this paper crucially rely on Lieb-Robinson bounds to prove that the heating time is at least exponential in $\omega$, we conjecture that the aforementioned gap stems from the lack of a tight Lieb-Robinson bound for $\alpha>D$, and we show the gap would vanish if such a tight bound were to exist. 
While the linear response theory also utilizes Lieb-Robinson bounds, it has weaker assumptions and, therefore, does not require a tighter bound to achieve the desired result of exponentially suppressed heating for $\alpha > D$. 

The remainder of the paper is organized as follows. 
In \cref{sec:Setup}, we provide definitions and describe the systems of interest. 
In \cref{sec:kLRB}, we review various Lieb-Robinson bounds for power-law systems and extend two of them\dash including one with the tightest light cone known to date\dash to $k$-body interactions. 
In \cref{sec:PRL}, we prove that in the linear response regime the heating time is at least exponential in $\omega$ for all $\alpha > D$. 
In \cref{sec:PRB}, we provide a more general analysis using the Magnus-like expansion and existing Lieb-Robinson bounds to prove exponentially-long heating times for $\alpha > 2D$.
We also conjecture a tight Lieb-Robinson bound that would extend this range of validity to $\alpha > D$. Finally, we summarize and discuss potential improvements in \cref{sec:Conclusion}.
\section{Setup and Definitions}\label{sec:Setup}
We consider a system of $N$ spins in a $D$-dimensional square lattice \footnote{While our results are derived considering a simple square lattice, we believe that it is not difficult to extend them to other regular lattices.}. 
The system evolves under a periodic, time-dependent Hamiltonian $H(t)$ with period $T$, i.e. $H(t+T) = H(t)$.
While the following analysis works for any $H(t)$ that is a sum of finite-body interactions, we assume that $H(t)$ is two-body for simplicity. 
Without loss of generality, we can write $H(t)=H_0+V(t)$ as the sum of a time-independent part $H_0$ and a time-dependent part $V(t)$ such that $\frac{1}{T}\int_0^T V(t) = 0$.
We further assume that $H_0$ and $V(t)$ are both power-law Hamiltonians with an exponent $\alpha$.

\begin{definition}\label{DEF_power-law}
	A Hamiltonian $H$ on a lattice $\Lambda$ is power-law with an exponent $\alpha$ and a local energy scale $\eta$ if we can write
	$H = \sum_X h_X,$
	where $h_X$ are Hamiltonians supported on subsets $X \subset \Lambda$, such that for any two distinct sites $i,j \in \Lambda$:
	\begin{align}
	\sum_{X\ni i,j} \norm{h_{X}}\leq \frac{\eta}{\dist(i,j)^\alpha},
	\end{align}
	and the norm $  \norm{h_{\{i\}}} \leq \eta$ for all $i \in \Lambda$, where $\norm{\cdot}$ denotes the operator norm and $\dist(i,j)$ the distance between sites $i,j$.
	In addition, we call $ \sup_X \abs{X}$ the local support size, where $\abs{X}$ is the number of sites in $X$, and define $\lnorm{H}= \sup_{i}\sum_{X\ni i}\norm{h_X}$ to be the local norm of $H$.
\end{definition}
In the following discussion, we assume $\eta  =1$, which sets the timescale for the dynamics of the system. 
In addition, we will occasionally write $H$ instead of $H(t)$ for brevity.
\section{Lieb-Robinson bounds for many-body power-law interactions} \label{sec:kLRB}
Before discussing the linear response theory and the Magnus-like expansion, it is helpful to review the existing Lieb-Robinson (LR) bounds for power-law interactions.
We will also generalize several of them from two-body to arbitrary $k$-body interactions for all $k\geq 2$. 
In particular, we discuss the relations between the bounds in Refs.~\cite{Hastings06,Gong14}, which imply logarithmic light cones for all $\alpha>D$, and the bounds in Refs.~\cite{Foss-Feig15,Else18,Tran18}, which imply algebraic light cones for $\alpha>2D$. 

\subsection{Lieb-Robinson bounds for $\alpha>D$}\label{sec:LRalpha>D}
First, we discuss the bounds in Refs.~\cite{Hastings06,Gong14}, which are valid for all $\alpha>D$.
Recall that LR bounds are upper bounds on the norm of the commutator $\comm{A(t),B}$, where $A,B$ are two operators supported on some subsets $X,Y$ of the lattice, and $A(t)$ is the time-evolved version of $A$ under a time-dependent Hamiltonian $H$.
The minimum distance between a site in $X$ and a site in $Y$ is $r = \dist(X,Y)>0$.
Since the sets $X,Y$ are disjoint, $\comm{A(0),B} = 0$ initially. 
As time grows, the operator $A(t)$ may spread to $Y$, making the commutator nontrivial.

The first LR bound for power-law interactions was proven in Ref.~\cite{Hastings06} by Hastings and Koma (HK):
\begin{align}
    \mathcal C(t,r)\equiv\norm{\comm{A(t),B}}\leq C \norm A \norm B \abs{X}\abs{Y} \frac{e^{vt}}{r^{\alpha}},\label{eq:LR-HK}
\end{align}
where $r = \dist(X,Y)$, $v\propto \eta$ is a constant that may depend on $\alpha$, and $C$ is a constant independent of the system. 
We shall also use the same $C$ to denote different inconsequential prefactors. 
Setting the commutator norm to a constant yields the light cone $t\gtrsim \log r$, which means it takes time at least proportional to $\log r$ for the commutator to reach a given constant value.

Technically, we can already use the HK bound in our later analysis of the heating time because it applies to $k$-body interactions for all $k$.
However, this bound is loose for large $\alpha$ for two reasons:
i) the velocity $v\propto 2^\alpha$ diverges for $\alpha\rightarrow \infty$,
ii) the light cone is logarithmic for all $\alpha$, which is unphysical since larger values of $\alpha$ correspond to shorter-range interactions and, therefore, we expect slower spreading of correlations.
In particular, we expect the light cone to become linear for large enough $\alpha$, given that the interactions are finite-range at $\alpha =\infty$.

Gong~\etal~\cite{Gong14} resolved the first challenge and derived a bound for two-body interactions:
\begin{align}
    \mathcal C(t,r)\leq C \norm A \norm B \abs{X}\abs{Y} 
    \left(
    \frac{e^{vt}}{[(1-\mu)r]^{\alpha}}
    +e^{vt-\mu r}
    \right)
    ,\label{eq:LR-ZX}
\end{align}
where $\mu\in (0,1)$ is an arbitrary constant.
The light cone implied by this bound is still logarithmic, but the velocity $v$ is finite for all $\alpha$.
Although the bound in Ref.~\cite{Gong14} was derived only for two-body interactions, their proof applies to arbitrary $k$-body interactions, where $k$ is a finite integer [See \Cref{sec:Gong} for a proof].

\subsection{Lieb-Robinson bounds for $\alpha>2D$}
In this section, we discuss the LR bounds for power-law interactions with $\alpha>2D$.
While the bounds in Ref.~\cite{Hastings06,Gong14} work for $\alpha>D$, they all have logarithmic light cones.
For $\alpha>2D$, it is possible to derive tighter algebraic light cones.
The first such bound was proven by Foss-Feig~\etal~\cite{Foss-Feig15} for two-body interactions (and generalized by Refs.~\cite{Matsuta2017,Else18} to $k$-body interactions for all $k\geq 2$).
A recent bound by Tran~\etal~\cite{Tran18}, however, gives a tighter algebraic light cone.
Here, we provide the generalization of that bound to $k$-body interactions:
\begin{align}
    \mathcal C(t,r)\leq &C \norm A \norm B  (r_0+r)^{D-1} \nonumber\\
    &\times\left(\frac{1}{(1-\mu)^\alpha}
    \frac{t^{\alpha-D}}{r^{\alpha-D-1}}
    +t  e^{-\frac{\xi r}{t}}
    \right)
    ,\label{eq:LR-Minh-constX}
\end{align}
where $r_0$ is the radius of the smallest ball that contains $X$ and $\mu, \xi\in(0,1)$ are arbitrary constants. 
The second term decays exponentially with $r/t$ and becomes negligible compared to the first term when $r\gg t$.
Note that, other than its dependence on $r_0$, this bound is independent of the size of $X,Y$ and is valid for $\alpha>2D$.

Before we present the proof of \cref{eq:LR-Minh-constX}, we summarize the key steps of the proof:
\begin{enumerate}
    \item First, divide $[0,t]$ into $M$ equal time intervals and define $t_0,t_1,\dots,t_M$ such that $t_{0} = 0$ and $t_{j+1}-t_j = \tau = t/M$. 
    We denote by $U_{t_{i},t_{j}}$ the evolution unitary of the system from time $t_i$ to $t_j$.
    \item Setting $U_j = U_{t_{M-j},t_{M-j+1}}$ for brevity, we can decompose the evolution of $A$ into $M$ timesteps:
    \begin{align}
        A(t) = U_M^\dag U_{M-1}^\dag\dots U_1^\dag A U_1 \dots U_{M-1} U_M.
    \end{align}
    \item We then use a truncation technique (explicitly described below) to approximate $U_1^\dag A U_1$ by some operator $A_1$ such that
    \begin{align}
        \norm{U_1^\dag A U_1 - A_1} = \epsilon_1,
    \end{align}
    and $A_1$ is supported on a ball of size at most $\ell$ larger than the size of the support of $A$.
    \item Repeat the above approximation for the other time slices, i.e. find $A_2,\dots,A_M$ such that 
    \begin{align}
        &\norm{U_2^\dag A_1 U_2 - A_2} = \epsilon_2,\\
        &\norm{U_3^\dag A_2 U_3 - A_3} = \epsilon_3,\\
                &\dots \nonumber\\
        &\norm{U_M^\dag A_{M-1} U_M - A_M} = \epsilon_M.
    \end{align}
    By the end of this process, we have approximated $A(t)$ by an operator $A_M$ whose support is at most $M\ell$ larger than the support of $A$.
    \item By choosing $M\ell$ just smaller than $r$, the support of $A_M$ does not overlap with the support of $B$. 
    Therefore, $\comm{A_M,B}=0$, and $\mathcal{C}(t,r)$ is at most the total error of the approximation, i.e.
    \begin{align}
        \epsilon = \epsilon_1+\dots+\epsilon_M.
    \end{align}
\end{enumerate}

The total error $\epsilon$, and hence the bound, depends on the truncation technique used in Step 3. 
In Ref.~\cite{Tran18}, the authors used a technique inspired by digital quantum simulation, which works for $\alpha>2D$.
However, in addition to truncating the evolution unitary, the technique in Ref.~\cite{Tran18} also truncates the Hamiltonian.
The large error from this truncation makes it difficult to further improve the bound.
Here, we use a different, simpler technique to generalize the bound in Ref.~\cite{Tran18} to $k$-body interactions for all $k\geq 2$.
Our technique does not require truncating the Hamiltonian, eliminating a hurdle for future improvements on the bound \footnote{We note that the approach in Ref.~\cite{Tran18} also gives the effective Hamiltonian that generates the evolution from $A$ to $A_M$, which is more useful than the technique presented here when knowing such a Hamiltonian is important, e.g. in digital quantum simulation.}.

Let us start without any assumption on the interactions of the system. 
We only assume that there already exists a bound on the commutator norm for the system:
\begin{align}
    \mathcal C (t,r) \leq f(t,r)\phi(X)\norm{A}\norm{B},\label{eq:old_LR}
\end{align}
for some function $f$ that increases with $t$ and decreases with $r$, where $\phi(X)$ is the boundary area of $X$. 

To truncate $U_1^\dag A U_1$, we simply trace out the part of $U_1^\dag A U_1$ that lies outside a ball of radius $\ell$ around the support of $A$~\cite{BravyiHV06}:
\begin{align}
    A_1 &\equiv\frac{1}{\Tr(\mathbb I_{\mathcal B_\ell(A)^c})} \Tr_{\mathcal B_\ell(A)^c} (U_1^\dag A U_1) \otimes \mathbb I_{\mathcal B_\ell(A)^c} \label{eq:traceA1} \\
    & = \int_{\mathcal B_\ell(A)^c} d\mu(W) W (U_1^\dag AU_1) W^\dag,\label{eq:traceint}
\end{align}
where $\mathcal B_\ell(A)$ is a ball of radius $\ell+r_0$ centered on $A$ and $X^c$ denotes the complement of the set $X$.
In \cref{eq:traceint}, we rewrite the trace over $\mathcal B_\ell(A)^c$ as an integral over the unitaries $W$ supported on $\mathcal B_\ell(A)^c$ and $\mu(W)$ denotes the Haar measure for the unitaries.
Effectively, $A_1$ is the part of $A$ that lies inside the ball $\mathcal B_\ell(A)$.
The error from approximating $U_1^\dag A U_1$ with $A$ is
\begin{align}
    &\epsilon_1=\norm{U_1^\dag A U_1 - A_1} \nonumber\\
    &= \norm{U_1^\dag A U_1 - \int_{\mathcal B_\ell(A)^c} d\mu(W) W (U_1^\dag AU_1) W^\dag}\nonumber\\
    &= \norm{\int_{\mathcal B_\ell(A)^c} d\mu(W) \left[U_1^\dag A U_1 - W (U_1^\dag AU_1) W^\dag\right]}\nonumber\\
    &\leq \int_{\mathcal B_\ell(A)^c} d \mu (W) \norm{\comm{U_1^\dag A U_1,W}}.
\end{align}
Note that $W$ is a unitary whose support is at least a distance $\ell$ from the support of $A$. 
Therefore, using the LR bound in \cref{eq:old_LR}, we have
\begin{align}
    \epsilon_1= \norm{U_1^\dag A U_1 - A_1}   &\leq \int_{\mathcal B_\ell(A)^c} d \mu (W) \norm{A} \phi(X) f(\tau,\ell)\nonumber\\
    &=\norm{A}  \phi(X)f(\tau,\ell) ,
\end{align}
where $\tau$ is the time interval of each time slice. 
In addition, it is clear from the definition of $A_{1}$ in \cref{eq:traceint} that $\norm{A_1}\leq \norm{A}$.
Therefore, the error of the approximation in the $j$-th time slice is at most
\begin{align}
    \epsilon_j \leq  \norm A \phi(X_{j-1}) f(\tau,\ell),
\end{align}
where $X_j$ is the support of $A_j$.
Thus, the new bound is
\begin{align}
    \mathcal C(t,r)\leq 2\norm{B}\epsilon &\leq 2 M \norm A \norm B  \phi_{\max}f(\tau,\ell)  \\
    &=2 \norm A\norm B \frac{t}{\tau} \phi_{\max} f(\tau,\ell),\label{EQ_bound_tl}
\end{align}
where $\phi_{\max}=\max_j \phi(X_j)$ and $M$ has been replaced by ${t}/{\tau}$.
Note that the above bound is valid for all choices of $t,\ell$, as long as
\begin{align}
    &M=\frac{t}{\tau}  < \frac{r}{\ell},\label{EQ_cond_1}\\
    &\ell\geq 1,\label{EQ_cond_2}\\
    &\tau \leq t. \label{EQ_cond_3}
\end{align}
The first condition ensures that the operator after the last time slice $A_M$ is still outside the support of $B$, while the last two are practical constraints. 

\Cref{EQ_cond_1} is equivalent to $\ell<r\tau/t$.
Because $f(\tau,\ell)$ is a decreasing function of $\ell$, the bound \cref{EQ_bound_tl} would be the tightest if we chose $\ell = \xi r\tau/t$ for some $\xi$ less than, but very close to, $1$.
The bound \cref{EQ_bound_tl} becomes
\begin{align}
    \mathcal C (t,r)
    &\leq  2\norm A\norm B \phi_{\max} \:f\left(\tau,\frac{\xi r \tau}{t}\right)  \frac{t}{\tau}  .\label{EQ_general_bound} 
\end{align}
Note that the only free parameter left is $\tau$, which is constrained by [see \cref{EQ_cond_1,EQ_cond_2,EQ_cond_3}]:
\begin{align}
    t \geq \tau > \frac{t}{r}.\label{EQ_cond_t}
\end{align}
We are now ready to generalize the bound in Ref.~\cite{Tran18} to many-body interactions.
Plugging the $k$-body generalization of \cref{eq:LR-ZX} [see \cref{eq:LR-ZX-noY-noX} in \cref{sec:Gong}] into \cref{EQ_general_bound}, we have
\begin{align}
    \mathcal C(t,r) &\leq C\norm A\norm B \phi_{\max} \frac{t}{\tau}\nonumber\\
    &\times\left(\frac{1}{(1-\mu)^\alpha}\frac{e^{v\tau}}{\left(\frac{ \xi r \tau}{t}\right)^{\alpha-D-1}}+e^{v\tau-\frac{\xi r \tau}{t}}\right)\nonumber\\
    &\leq C\norm A\norm B (r_0+r)^{D-1} \frac{t}{\tau}\nonumber\\
    &\times\left(\frac{1}{(1-\mu)^\alpha}\frac{e^{v\tau}}{\tau^{\alpha-D-1}}\left(\frac{t}{ r }\right)^{\alpha-D-1}+e^{v\tau-\frac{\xi r \tau}{t}}\right),\nonumber
\end{align}
where we have assumed without loss of generality that $X$ is a ball of radius $r_0$ and replaced $\phi_{\max}\propto (r_0+r)^{D-1}$.
Taking $\tau=1$ to be a constant, we obtain a bound that is valid for all $\alpha>D+1$:
\begin{align}
    \mathcal C(t,r)  
    &\leq C\norm A\norm B (r_0+r)^{D-1} \nonumber\\
    &\quad\quad\times\bigg(\frac{1}{(1-\mu)^\alpha}\frac{t^{\alpha-D}}{r^{\alpha-D-1}}+ te^{-\frac{\xi r}{t}}\bigg).\nonumber
\end{align}
In particular, if $r_0$ is a constant, we can simplify (in the limit of large $t,r$) to
\begin{align}
    \mathcal C(t,r) 
    &\leq C \norm A\norm B  \left(\frac{1}{(1-\mu)^\alpha}\frac{t^{\alpha-D}}{r^{\alpha-2D}}+tr^{D-1}e^{-\frac{\xi r}{t}}\right).\label{eq:Minh-r0=1}
\end{align}
Note that although the bound is, in principle, valid for $\alpha>D+1$, it is only useful for $\alpha>2D$.

\section{Linear response theory}\label{sec:PRL}
In this section, we present the derivation of an exponentially suppressed heating rate for periodically driven, power-law Hamiltonians under the assumptions of linear response theory.
We will assume that the drive $V(t)$ is harmonic and local.
That is, we can write $V(t) = g \cos(\omega t) O$, for some small constant $g$ and some time-independent operator $O= \sum_i O_i$ composed of local operators $O_i$.
For simplicity, we assume each $O_i$ is supported on a single site $i$ (but our results also hold when $O_i$ is supported on a finite number of sites around $i$). 
We also assume the system is initially in a thermal state $\rho_\beta$ of $H_{0}$ with a temperature $\beta^{-1}$.
Within the linear response theory, the energy absorption rate is proportional to the dissipative (imaginary) part of the response function $\sigma(\omega) = \sum_{i,j}\sigma_{ij}(\omega)$~\cite{Abanin15}, where 
	\begin{align}
		\sigma_{ij}(\omega) &= \frac{1}{2}\int_{-\infty}^{\infty}dt e^{i\omega t}\avg{\left[ O_{i}(t), O_{j}(0) \right]}_{\beta}, 
	\end{align}
$\avg{O}_\beta \equiv \Tr(\rho_\beta O)$ denotes the thermal average of $O$, and $O(t) = e^{iH_0t}Oe^{-iH_0t}$ is the time-evolved version of $O$ under $H_0$.

The authors of Ref.~\cite{Abanin15} showed that there exists a constant $\kappa$ such that for all $i,j$ and for all $\omega,\delta\omega>0$, the $(i,j)$ entry of $\sigma(\omega)$ can be bounded as
\begin{equation}\label{EQN:LR_Exp_bound}
	\abs{\sigma_{ij}([\omega, \omega + \delta\omega])} \leq e^{-\kappa \omega},
\end{equation}
where $f([\omega_1,\omega_2])\equiv\int_{\omega_1}^{\omega_2} f(\omega)d\omega$.
Although the statement of Ref.~\cite{Abanin15} applies to Hamiltonians with finite-range interactions, we show in \cref{sec:sigmaij_bound_proof} that it also holds for power-law Hamiltonians for all $\alpha \geq 0$.

In principle, \cref{EQN:LR_Exp_bound} already implies that the absorption rate of a finite system is exponentially small as a function of the frequency $\omega$.
However, since there are $N$ sites in the system, naively applying \cref{EQN:LR_Exp_bound} by summing over the indices $i,j$ yields a superextensive heating rate $\sim N^{2}e^{-\kappa \omega}$.
Such superextensivity is non-physical, as it would imply that a local drive instigates a divergent absorption per site in the thermodynamic limit. 
To address this, Ref.~\cite{Abanin15} introduced a bound complementary to \cref{EQN:LR_Exp_bound}\dash based on Lieb-Robinson bounds for finite-range interactions~\cite{LR}\dash that implies the contribution from the off-diagonal terms is also \emph{exponentially} suppressed with the distance $r_{ij}$ between the sites~$i,j$.

The case of power-law interacting Hamiltonians is somewhat more involved.
Due to the long-range interaction, the commutator $\avg{\comm{O_i(t),O_j}}_\beta$ can decay more slowly as a function of $r_{ij}$ than in the finite-range case.
Fortunately, we show that it still decays sufficiently quickly for us to recover the extensive, exponentially-small heating rate for power-law Hamiltonians.
We provide the technical proof in \cref{sec:linear-response-bound-proof}, but a high-level argument goes as follows.

Lieb-Robinson bounds for power-law systems with $\alpha > D$~\cite{Hastings06,Gong14,Else18,Tran18} imply that the contributions from the $(i,j)$ entries are suppressed by $1/r_{ij}^\alpha$.
Therefore, the total contribution to $\sigma([\omega,\omega+\delta\omega])$ from the pairs $(i,j)$ with $r_{ij}$ larger than some distance $r_*$ (to be chosen later) is at most
\begin{align}
	\sum_{i,j:r_{ij}\geq r_*}\frac{C}{r_{ij}^{\alpha}}
	\leq \frac{CN}{r_*^{\alpha-D}}, \label{eq:sigma_bound_out}
\end{align} 
where we use the same notation $C$ to denote different constants that are independent of $r_{ij}, t, $ and $N$.
The factor $N$ comes from summing over $i$ and the factor $r_*^D$ from summing over $j$ at least a distance $r_*$ from $i$.

For $r_{ij}\leq r_*$, we simply use the bound in \cref{EQN:LR_Exp_bound} to bound their contributions:
\begin{align}
	\sum_{i,j:r_{ij}\leq r_*} C e^{-\kappa\omega} 
	\leq C N r_*^D e^{-\kappa \omega},\label{eq:sigma_bound_in}
\end{align}
where $Nr_*^D$ is roughly the number of pairs $(i,j)$ separated by distances less than $r_*$.
Combining \cref{eq:sigma_bound_out} with \cref{eq:sigma_bound_in}, we get
$
	\abs{\sigma([\omega,\omega+\delta\omega])}
	\leq CNr_*^D\left(e^{-\kappa \omega}+r_*^{-\alpha}\right).
$
Finally, choosing $r_* = \exp(\kappa\omega/\alpha)$, we obtain a bound on the absorption rate,
\begin{align}
	\abs{\sigma([\omega,\omega+\delta\omega])}
	\leq CN\exp\left[-\left(1-\frac{D}{\alpha}\right)\kappa\omega\right],\label{eq:linear-response-bound}
\end{align}
which decays exponentially quickly with $\omega$ as long as $\alpha>D$.
Thus, we have shown that, within the linear response theory, the heating rate of power-law interacting Hamiltonians obeys a bound that is qualitatively similar to that for finite-range interactions: the heating rate is extensive, but exponentially small in the driving frequency.

\section{Magnus-like expansion}\label{sec:PRB}
We now present a more general approach to proving a bound on the heating time in a system governed by a periodically driven, power-law Hamiltonian. 
In particular, this approach remains correct for strongly driven systems, where linear-response theory does not apply. 
We generalize Ref.~\cite{Abanin17PRB} and construct an effective time-independent Hamiltonian $H_*$.
The leading terms of $H_*$ resemble the effective Hamiltonian one would get from the Magnus expansion~\cite{Blanes2009,Bukov2015,Eckardt2015}. 
Using Lieb-Robinson bounds for power-law interactions, we show that the evolution of local observables under $H_*$ well approximates the exact evolution up to time $t_*$, which is exponentially long as a function of the drive frequency.
Additionally, the existence of the effective Hamiltonian $H_*$ also implies a prethermalization window during which the system could thermalize with respect to $H_*$ before eventually heating up after time $t_*$.

Following Ref.~\cite{Abanin17PRB}, we construct a periodic unitary transformation $Q(t)$ such that $Q(t+T) = Q(t)$ and $Q(0) = \mathbb I$.
After moving into the frame rotated by $Q(t)$, we show that the transformed Hamiltonian is nearly time-independent and the norm of the residual time-dependent part is exponentially small as a function of the frequency.

To construct the unitary $Q(t)$, we note that the state of the system in the rotated frame, $\ket{\phi(t)} = Q^{\dag}(t)\ket{\psi(t)}$, obeys the Schr\"{o}dinger equation with a transformed Hamiltonian $H'(t)$ ($\hbar = 1$):
\begin{equation}
i \partial_{t} \ket{\phi(t)}  
= (Q^{\dag}HQ - iQ^{\dag}\partial_{t}Q)\ket{\phi(t)}
\equiv H'(t)\ket{\phi(t)}.\label{eq:H'def}
\end{equation}
We write $Q = e^\Omega$, where $\Omega(t)$ is a periodic, anti-Hermitian operator, i.e. $\Omega(t) = \Omega(t+T)$ and $\Omega^\dag = -\Omega$. 
We then assume that the period $T$ is small so that we may expand $\Omega(t) = \sum_{q=1}^{\infty} \Omega_{q}$ in orders of $T$, where $\norm{\Omega_q}=\O{T^q}$, and we will eventually choose $\Omega_{q}$ such that the transformed Hamiltonian $H'(t)$ is almost time-independent. 
In particular, we shall truncate the expansion of $\Omega(t)$ up to order $\qmax$ and choose $\Omega_q$ recursively for all $q\leq\qmax$ to minimize the norm of the driving term in $H'(t)$. 

We can rewrite $H'(t)$ from \cref{eq:H'def} as:
\begin{equation}
H'(t) = e^{-\ad_{\Omega}}[H_{0}+V(t)] - i\frac{1-e^{-\ad_{\Omega}}}{\ad_{\Omega}}\partial_{t}\Omega 
, \label{eq:H'sum}
\end{equation}
with \(\ad_{\Omega}A = [\Omega,A]\). 
From \cref{eq:H'sum}, we can define $H'_q(t)$ for $q=0,1,\dots$ such that $H' = \sum_{q = 0}^{\infty}H'_q(t)$ is expanded in powers of $T$:
\begin{align}
&H'_{q}(t) = G_{q}(t) - i\partial_{t}\Omega_{q+1}(t) ,  \label{eq:Hq} 
\end{align}
where we define $G_q$ via $\Omega_1,\dots,\Omega_q$ as follows:
{
\medmuskip=-0mu
\thinmuskip=-0mu
\thickmuskip=-0mu
\nulldelimiterspace=0pt
\scriptspace=0pt
\begin{align}	
&G_{q}(t) 
= 
\sum_{k=1}^{q}\frac{(-1)^{k}}{k!}\sum_{\substack{1 \leq i_{1}, \dots, i_{k} \leq q \\ i_{1} + \dots + i_{k} = q}} \ad_{\Omega_{i_{1}}}\dots \ad_{\Omega_{i_{k}}} H    (t)  
\nonumber\\
&+ i\sum_{k=1}^{q}\frac{(-1)^{k+1}}{(k+1)!} \sum_{\substack{1\leq i_{1}, \dots, i_{k},m \leq q+1 \\ i_{1} + \dots + i_{k} + m = q+1}} \ad_{\Omega_{i_{1}}}\dots\ad_{\Omega_{i_{k}}}\partial_{t}\Omega_{m} ,
 \label{eq:Gq}
\end{align}
\medmuskip=1mu
\thinmuskip=1mu
\thickmuskip=1mu
\nulldelimiterspace=1pt
\scriptspace=1pt
and $G_0(t) = H(t)$.}
Now, recall that $\Omega_{q}(t)$ are operators that we can choose.
From \cref{eq:Hq}, we choose $\Omega_1(t)$ such that it cancels out the time-dependent part of $G_0(t)$, making $H'_0$ time-independent.
This choice of $\Omega_1(t)$ also defines $G_1(t)$.
We then choose $\Omega_2(t)$ to eliminate the time-dependent part of $G_1(t)$. 
In general, we choose $\Omega_q$ successively from $q=1$ to some $q=\qmax$ (to be specified later) so that $H'_q$ are time-independent for all $q<\qmax$.
Therefore, the remaining time-dependent part of the transformed Hamiltonian $H'(t)$ must be at least $\mathcal{O}(T^{\qmax})$.
Specifically, for $q<\qmax$, we choose the following:
\begin{align}
\bar{H}'_{q} &= \frac{1}{T}\int_{0}^{T}G_{q}(t)dt, \label{Hqbar}\\
\Omega_{q+1}(t) &= -i\int_{0}^{t}\left(G_{q}(t')-\bar{H}'_{q}\right)dt'.      \label{Omegaq}
\end{align}
Here, \cref{Omegaq} ensures that \cref{eq:Hq} becomes $H'_q(t) = \bar H'_q$, and, thus, that $H'_q$ is time-independent for all $q<\qmax$.
On the other hand, for $q\geq \qmax$, we choose $\Omega_{q+1}(t) = 0$, so that
$H'_q(t) = G_q(t)$.
By this construction, we can rewrite the transformed Hamiltonian into the sum of a time-independent Hamiltonian $H_*$ and a drive $V'(t)$ that contains higher orders in $T$:
\begin{align}
H'(t) &=\sum_{q=0}^\infty H_q(t) = \underbrace{\sum_{q=0}^{\qmax-1}\bar{H'}_{q}}_{\equiv H_*} + \underbrace{\sum_{q = \qmax}^{\infty}G_{q}(t)}_{\equiv V'(t)}, \label{eq:Hprime}
\end{align}

As a result of the transformation, the driving term $V'(t)$ is now $\mathcal{O}(T^{\qmax})$.
As discussed before, we will eventually choose the cutoff $\qmax$ to minimize the norm of the residual drive $\norm{V'(t)}$.

To estimate the norm of $V'(t)$, elucidating its dependence on $\qmax$, we first need more information on the structure of the $\Omega_q(t)$ for all $1\leq q \leq \qmax$. 
In particular, we show that $G_q$ and $\Omega_{q}$ are power-law interacting Hamiltonians.
To do so, we first need to define some more notation. 
We denote by $\Hpl$ the set of power-law Hamiltonians with exponent $\alpha$ and a local energy scale $\eta = 1$.
In addition, we denote by $\Hpl^{(k)}$ the subset of $\Hpl$ which contains all power-law Hamiltonians whose local support size [see \cref{DEF_power-law}] is at most $k+1$.
For a real positive constant $a$, we also denote by $a\Hpl$ the set of Hamiltonians $H$ such that $a^{-1} H$ is a power-law Hamiltonian with the same exponent $\alpha$.

The following lemma says that $G_q$ and $\Omega_{q}$ are also power-law Hamiltonians up to a prefactor.
\begin{lemma}\label{LEM_Gq}
	For all $q< \qmax$, we have
	\begin{align}
		&G_q\in T^q q! c^{q} \lambda^q \Hpl^{(q+1)},\label{eq:GqBound}\\
		&\partial_t \Omega_{q+1} \in T^{q}q! c^{q}\lambda^q\Hpl^{(q+1)}, \label{eq:dtOmegaBound}\\
		&\Omega_{q+1} \in T^{q+1}q! c^{q}\lambda^q\Hpl^{(q+1)}, \label{eq:OmegaBound}
	\end{align}
	where $c,\lambda$ are constants to be defined later.
\end{lemma}
Observe that for any order $q$, the last two bounds, i.e. \cref{eq:dtOmegaBound} and \cref{eq:OmegaBound}, follow immediately from \cref{eq:GqBound} and the definition of $\Omega_q$.
Note that \cref{LEM_Gq} holds for $G_0(t) = H(t)\in \Hpl^{(1)}$.
It is also straightforward to prove \cref{LEM_Gq} inductively on $q$.
The factor $T^q$ comes from the constraint in \cref{eq:Gq} that $i_1+\dots+i_k=q$, along with the fact that each $\Omega_{i_\nu}$ is $\mathcal{O}(T^{i_\nu})$ for all $\nu=1,\dots,k$.
Similarly, the factor of $q!$ is combinatorial and comes from the nested commutators in \cref{eq:Gq}.
We provide a more technical proof of \cref{LEM_Gq} in \cref{sec:Gq}.

As a consequence of \cref{LEM_Gq}, we can bound the local norms of the operators:

\begin{align}
	\norm{G_q}_l &\leq T^q q! c^{q} \lambda^{q+1} \leq \lambda e \sqrt{q}\left( \frac{Tqc\lambda}{e} \right)^{q}, \label{GqNormBound}\\
	\norm{\Omega_q}_l &\leq T^{q} (q-1)! c^{q-1} \lambda^q
	\leq \frac{e}{c} \left(\frac{Tqc\lambda}{e}\right)^{q}. \label{OmegaqNormBound}
\end{align}
There are two competing factors in the bounds: $T^q$, which decreases with $q$, and $q!\sim q^{q}$, which increases with $q$.
This suggests that the optimal choice for $\qmax$\dash in order to minimize the local norm in \cref{GqNormBound}\dash should be around ${e/(cT\lambda)}$. 
In the following, we shall choose
\begin{align}\label{qmax}
	\qmax = \nmax \equiv \frac{e}{cT\lambda} e^{-\kappa},
\end{align}
for some $\kappa > \ln 2$. 
Note also that $\nmax$ is equal to frequency $\omega=1/T$ up to a constant.
With this choice of $\qmax$, \cref{eq:GqBound} reduces to 
\begin{align}
	G_q\in \lambda e \sqrt{q} e^{-\kappa q} \Hpl^{(q+1)},\label{eq:Gq<expsmall}
\end{align}
for all $q< \qmax =\nmax$.
By summing over $G_q$ with $q<\nmax$, we find that the effective time-independent Hamiltonian $H_{*}$ [see \cref{eq:Hprime}] is also a power-law Hamiltonian, i.e. 
$H_{*} \in C \Hpl^{(\qmax)}\in C \Hpl,$
up to a constant $C$ that may depend only on $\kappa$.

Similarly, we find from \cref{eq:OmegaBound} that $\Omega_q\in e/(c\lambda) e^{-\kappa q} \Hpl^{(q)}$ for all $q\leq \nmax$.
Plugging into the definition of $G_q$ and noting that we choose $\Omega_q=0$ for all $q\geq \qmax$, we find an identity similar to \cref{eq:Gq<expsmall}, but for $q\geq \nmax$:
\begin{align}
G_q\in C e^{-\kappa' q} \Hpl,\label{eq:Gq>expsmall}
\end{align}
where $\kappa'>\kappa-\ln 2$ is a constant.
Summing over $G_q$ with $q\geq\qmax$ [see \cref{eq:Hprime}], we again find that the residual drive $V'(t)$ is a power-law Hamiltonian up to a prefactor that decays exponentially with $\nmax$:
	\begin{align}
		V'(t) \in C e^{-\kappa'\nmax} \Hpl,\label{eq:V'}
	\end{align}
	where $C$ and $\kappa'$ are some positive constants.  
As a result, the local norm of $V'(t)$ decreases exponentially with $\nmax$:
$
	\norm{V'(t)}_l\leq C\lambda e^{-\kappa' \nmax}.
$

As discussed earlier, \cref{eq:Hprime} and \cref{eq:V'} imply the existence of an effective time-independent Hamiltonian $H_*$ such that the difference $\norm{Q^\dag H Q - H_*} = \norm{V'}$ is exponentially small as a function of $\nmax \propto 1/T$. 
However, even if $\norm{V'}_{l}$ is exponentially small, $\norm{V'}$ still diverges in the thermodynamic limit.
Therefore, in order to characterize the heating rate of the Hamiltonian, it is necessary to investigate the evolution of a local observable $O$ under $H(t)$.
We show that the evolution is well described by the effective time-independent Hamiltonian $H_*$ at stroboscopic times $t = T\mathbb{Z}$.
Without loss of generality, we assume the local observable $O$ is supported on a single site and $\norm O= 1$.
Following a similar technique used in Abanin~\etal~\cite{Abanin17PRB}, we write the difference between the approximate evolution under the effective Hamiltonian and the exact evolution (in the rotated frame):
\begin{align}
	\delta &= Q(t) U^\dag(t) O U(t) Q^\dag(t) - e^{it H_{*}} O e^{-it H_{*}}\nonumber\\
	&=i \int_0^t ds W^\dag(s,t)\comm{V'(s),e^{isH_{*}}Oe^{-isH_{*}}}W(s,t),\nonumber
\end{align}
where $U(t) = \mathcal{T}\exp\left( {-i\int_{0}^{t}H(t')dt'} \right)$ is the time evolution generated by the full Hamiltonian $H(t)$ and $W(s,t)=\mathcal{T}\exp\left({-i\int_{s}^{t}H'(t')dt'} \right)$ is the evolution from time $s$ to $t$ generated by $H'(t)$.
We can then bound the norm of the difference using the triangle inequality:
\begin{align}
	\norm{\delta}\leq \int_0^t ds \norm{\comm{V'(s),e^{isH_{*}}Oe^{-isH_{*}}}}.\label{eq:deltanorm}
\end{align}
We can bound the right-hand side of \cref{eq:deltanorm} using Lieb-Robinson bounds for power-law interactions.

First, we provide an intuitive explanation why the norm of $\delta$ is small for small time.
Recall that the operator $O$ is initially localized on a single site.
At small time, it is still quasilocal and therefore significantly noncommutative with only a small number of terms of $V'$ lying inside the ``light cone'' generated by the evolution under $H_{*}$.
There are several Lieb-Robinson bounds for power-law interactions~\cite{Gong14,Foss-Feig15,Tran18,Else18} [see also \cref{eq:LR-ZX-noY-noX} and \cref{eq:LR-Minh-constX}], each provides a different estimate for the shape of the light cone, resulting in a different bound for the heating time. 

If the light cone is logarithmic (as bounded in Ref.~\cite{Gong14}), the commutator norm in \cref{eq:deltanorm} would grow exponentially quickly with time and eventually negate the exponentially-small factor $\exp(-\kappa' \nmax)$ from $\lnorm{V'}$.
Therefore, in such cases, the system could potentially heat up only after $t_*\propto \nmax = 1/T$.
On the other hand, if we use the Lieb-Robinson bounds that imply algebraic light cones (as in Refs.~\cite{Foss-Feig15,Tran18,Else18} for $\alpha > 2D$), the commutator norm only grows subexponentially with time, and we can expect to recover the exponentially-long heating time $t_*\propto e^{\kappa' \nmax}$ derived for finite-range interactions~\cite{Abanin17PRB,Abanin17CMP}.

\cref{sec:applylr} contains the mathematical details, but the results of this analysis are as follows. Using Gong~\etal~\cite{Gong14} [or its $k$-body generalization \cref{eq:LR-ZX-noY}], which holds for $\alpha>D$ and has a logarithmic light cone $t\gtrsim \log r$, yields:\begin{align}
        \norm{\delta}\leq C e^{-\kappa' \nmax} e^{2Dvt/\alpha}.
\end{align}
Thus, the difference $\delta$ is only small for time $t_*\propto \nmax\propto 1/T$. 
This behavior is expected because the region inside the light cone implied by Gong~\etal~'s bound expands exponentially quickly with time. 

If instead we use the bound in Else~\etal~\cite{Else18}, we find:
\begin{align}
        \norm{\delta}\leq C e^{-\kappa' \nmax} \xi\left(\frac{D}{1-\sigma}\right)  t^{\frac{D}{1-\sigma}+1},\label{eq:usingElse}
\end{align}
where $\xi(x) \equiv \frac{1}{x}2^x \Gamma(x) $ and $\Gamma$ is the Gamma function. 
Thus, the difference is small up to an exponentially long time $t_*\propto e^{\kappa' \nmax \frac{1-\sigma}{D +1 - \sigma}}$.
The result holds for $\alpha > D\left(1+\frac{1}{\sigma}\right)$, where $\sigma$ can be chosen arbitrarily close to 1.
This condition is effectively equivalent to $\alpha>2D$ [see \cref{sec:applylr} for a discussion of the limit $\sigma \to 1^-$].

We may also use the bound in Tran~\etal~\cite{Tran18} [see \cref{eq:Minh-r0=1} for its generalization to $k$-body interactions], which gives
\begin{align}
        \norm{\delta}\leq C e^{-\kappa' \nmax} t^{\frac{D(\alpha-D)}{\alpha-2D}+1}.
\end{align}
Thus, the difference is small up to an exponentially-long time $t_*\propto \exp\left({\kappa' \nmax\frac{\alpha-2D}{\alpha(D+1)-D(D+2)}}\right)$.
This analysis works only when $\alpha>3D$, but, within this regime, the exponent of the heating time using this bound is larger than obtained in \cref{eq:usingElse}.
This is due to the trade-off between the tail and the light cone between the bounds in Refs.~\cite{Else18,Tran18}. 
See \cref{sec:applylr} for more details.

Finally, we conjecture a tight bound for power-law interactions that holds for all $\alpha > D$, and we will provide the full derivation of $\delta$ for such a bound.
First, we consider the light cone of such a bound. 
Given the best known protocols for quantum information transfer~\cite{Zach17}, the best light cone we could hope for would be $t\gtrsim r^{\alpha-D}$ for $D+1>\alpha>D$ and linear for $\alpha>D+1$. 
In the following, we assume the light cone of the conjectured bound is $t\gtrsim r^{1/\beta}$ for some constant $\beta\geq 1$ for all $\alpha>D$.

Next, we consider the tail of the bound, i.e. how the conjectured bound decays with the distance at a fixed time. 
Since it is always possible to signal between two sites using their direct interaction, which is of strength $1/r^\alpha$, the tail of the bound cannot decay faster than $1/r^\alpha$. 
We shall assume that the bound decays with the distance exactly as $1/r^\alpha$.

For simplicity, we assume that the conjectured bound takes the form
\begin{align}
    \norm{\comm{A(t),B}}\leq C\norm A\norm B \left(\frac{t^{\beta}}{r}\right)^{\alpha},
\end{align}
which manifestly has a light cone $t\gtrsim r^{1/\beta}$ and decays as $1/r^{\alpha}$ with the distance. 
Let $r_*(t) = t^{\beta}$ be the light cone boundary and consider the sum inside and outside the light cone.

For convenience, denote $V'' = C^{-1} e^{\kappa' \nmax} V', \bar H'' = \gamma^{-1} H_*$ so that $V'',\bar H''\in \Hpl$.
We can rewrite the bound on $\norm{\delta}$ as
\begin{align}
\norm{\delta}\leq C e^{-\kappa' \nmax}\int_0^t ds \norm{\comm{V''(s),e^{is\gamma \bar H''}Oe^{-is\gamma\bar H''}}},\label{EQ_delta_int}
\end{align}
Now write $V''(s) = \sum_{r=0}^\infty V''_r(s)$, where $V''_r(s)\equiv \sum_{X:\dist(X,O)\in[r,r+1)} h_X$ denotes the terms of $V''(s)$ supported on subsets exactly a distance between $r$ and $r+1$ away from $O$.
Since 
$V''(s)$ is a power-law Hamiltonian, it follows that $\norm{V''_r(t)}\leq C r^{D-1}$.
Writing the sum this way, we can now separate terms inside and outside of the light cone. 

For the terms inside the light cone, we bound:
\begin{align}
        &\sum_{r\leq r_*(s)} \norm{\comm{V''_r(s),e^{is\gamma \bar H''}Oe^{-is\gamma\bar H''}}}
        \nonumber\\
        &\leq 2\sum_{r\leq r_*(s)}  \norm{V''_r(s)}\norm {O}
        \leq C r_*(s)^{D}
        \leq C s^{\beta D}.\label{eq:usingtight-in}
\end{align}

For the terms outside the light cone, we use the conjectured bound:
\begin{align}
        &\sum_{r>r_*(s)} \norm{\comm{V''_r(s),e^{is\gamma \bar H''}Oe^{-is\gamma\bar H''}}}\nonumber\\
        &\leq C \sum_{r> r_*(s)}  \norm{V''_r(s)}\norm {O} \frac{s^{\beta\alpha}}{r^{\alpha}}\\
        &\leq C  \sum_{r> r_*(s)} \frac{s^{\beta\alpha}}{r^{\alpha-D+1}}
        \leq C  \frac{s^{\beta\alpha}}{r_*(s)^{\alpha-D}} = C s^{\beta D}.\label{eq:usingtight-out}
\end{align} 

Combining \cref{eq:usingtight-in} and \cref{eq:usingtight-out}, we get
\begin{align}
    \norm{\delta}\leq C e^{-\kappa' \nmax} t^{\beta D+1},
\end{align} 
which implies an exponential heating time as a function of $\nmax$, i.e. $t_*\propto \exp(\kappa' \nmax/{(\beta D+1)})$. 
Recall that the best values we can hope for $\beta$ are $\beta = 1/(\alpha-D)$ when $D+1>\alpha>D$ and $\beta = 1$ when $\alpha>D+1$.
Note also that the exponential heating time would hold for all $\alpha>D$, matching the result given by the linear response theory.

\section{Conclusion \& outlook}\label{sec:Conclusion}
Our work generalizes the results of Refs.~\cite{Abanin15,Abanin17PRB,Abanin17CMP} for finite-range interactions to power-law interactions.
Using two independent approaches, we show that periodically driven, power-law systems with a large enough exponent $\alpha$ can only heat up after time that is exponentially long in the drive frequency. 
The results only hold if $\alpha$ is larger than some critical value $\alpha_c$. 
Physically, the existence of $\alpha_c$ coincides with our expectation that power-law interactions with a large enough exponent $\alpha$ are effectively short-range.

However, the two approaches imply different values for $\alpha_c$.
While both the Magnus expansion in Ref.~\cite{Kuwahara2016} and the Magnus-like expansion in this paper independently suggest $\alpha_c = 2D$, the linear response theory implies $\alpha_c = D$.
We conjecture that this gap is due to the lack of tighter Lieb-Robinson bounds for power-law interactions, especially for $\alpha$ between $D$ and $2D$.
Indeed, we demonstrated in \cref{sec:PRB} that a tight Lieb-Robinson bound for this range of $\alpha$ implies an exponentially-long heating time for all $\alpha>D$, matching the result from the linear response approach, as well as previous numerical evidence for systems with $\alpha<2D$~\cite{Machado2017}.
Therefore, proving a tight Lieb-Robinson bound has important implications for the heating time of power-law interacting systems.

\textit{Note added.}---During the preparation of this manuscript, we became aware of a related complementary work on long-range prethermal phases \cite{machado19}.
We also became aware of a tighter Lieb-Robinson bound for power-law interactions~\cite{Chen19}. 
However, the bound has a range of validity $\alpha>2$ in one dimension and, thus, does not close the aforementioned gap. \begin{acknowledgments}
	We thank A.~Deshpande, J.~R.~Garrison, Z.~Eldredge, and Z-X.~Gong for helpful discussions. 
	MCT, AE, AYG, and AVG acknowledge funding from the DoE ASCR Quantum Testbed Pathfinder program (Award No. DE-SC0019040), the NSF PFCQC program, DoE BES Materials and Chemical Sciences Research for Quantum Information Science program (Award No. DE-SC0019449),
	NSF PFC at JQI, ARO MURI, ARL CDQI, and AFOSR\@.
	MCT acknowledges the NSF Grant No.~PHY-1748958 and the Heising-Simons Foundation.
	AVG acknowledges funding from NSF under grant No. Phy-1748958.
	AE acknowledges funding from the DoD. 
	AYG is supported by the NSF Graduate Research Fellowship Program under Grant No.\ DGE 1322106.
	PT was supported by the NIST NRC Research Postdoctoral Associateship Award.
	DAA acknowledges support from the Swiss National Science Foundation.
\end{acknowledgments}
\bibliographystyle{apsrev4-1}
\bibliography{heating-magnus}

\appendix

\section{Generalization of Gong~\etal~\cite{Gong14} to many-body interactions}
\label{sec:Gong}
In this section, we prove \cref{eq:LR-ZX} and thereby generalize the bound in Gong~\etal~\cite{Gong14} from two-body to $k$-body interactions, where $k$ is an arbitrary finite integer. 
This bound is an ingredient in the generalization of the tighter Lieb-Robinson bound in Tran~\etal~\cite{Tran18} to $k$-body interactions.

\begin{proof}
We recall that the bound in Ref.~\cite{Gong14} is based on the Hastings \& Koma series~\cite{Hastings06}:
\begin{widetext}
\begin{equation}\label{eq:HK0}
\|[A(t),B]\|\le2\|A\|\|B\|\sum_{k=1}^{\infty}\frac{(2t)^{k}}{k!}\left[\sum_{Z_{1}:Z_{1}\cap X\neq\varnothing}\sum_{Z_{2}:Z_{1}\cap Z_{2}\neq\varnothing}\dots\sum_{\substack{Z_{k}:Z_{k-1}\cap Z_{k}\neq\varnothing,\\
Z_{k}\cap Y\neq\varnothing
}
}\prod_{i=1}^{k}\|h_{Z_{i}}\|\right],
\end{equation}
and we can bound the summation within the
square brackets as
\begin{align}
\sum_{Z_{1}:Z_{1}\cap X\neq\varnothing}\sum_{Z_{2}:Z_{1}\cap Z_{2}\neq\varnothing}\dots\sum_{\substack{Z_{k}:Z_{k-1}\cap Z_{k}\neq\varnothing,\\
Z_{k}\cap Y\neq\varnothing
}
}\prod_{i=1}^{k}\|h_{Z_{i}}\| 
& \leq\sum_{i\in X}\sum_{j\in Y}\sum_{z_{1}}\sum_{z_{2}}\dots\sum_{z_{k-1}}
\bigg(\sum_{Z_{1}\ni i,z_{1}}\norm{h_{Z_1}}\bigg)
\dots
\bigg(
\sum_{\substack{Z_{k}\ni z_{k-1},j}}
\norm{h_{Z_k}}
\bigg)
\nonumber\\
& \leq\sum_{i\in X}\sum_{j\in Y}\lambda^k\mathcal{J}^{k}(i,j),\label{eq:HK} 
\end{align}
\end{widetext}
where $\mathcal{J}^{k}(i,j)$ is given by the $k$-fold convolution
of the hopping terms $J_{ij} \equiv\frac{1}{r_{ij}^\alpha}$ (where $r_{ij} = \dist(i, j)$) for $i\neq j$ and $J_{ii} = 1$ for all $i$:
\begin{align*}
\mathcal{J}^{k}(i,j) & \equiv\sum_{z_{1}}\sum_{z_{2}}\dots\sum_{z_{k-1}}
J_{iz_1}J_{z_{1}z_{2}}\dots J_{z_{k-1}j}. \label{eq:Jconvo}
\end{align*}
Note that \cref{eq:HK} comes from \cref{DEF_power-law}: $\sum_{Z\ni i,j}\norm{h_Z}\leq 1/r_{ij}^\alpha = J_{ij}$ for $i\neq j$ and
\begin{align}
\sum_{Z\ni i}\norm{h_Z}\leq \sum_{j}\sum_{Z\ni i,j} \norm{h_Z} \leq \lambda ,
\end{align}
where $\lambda =1+\sum_{j\neq i} 1/r_{ij}^\alpha$ is a finite constant for all $\alpha>D$.
This equation is exactly Eq.~(3) in Ref.~\cite{Gong14}.

For simplicity, we consider $D=1$ in the following discussion. 
To put a bound on $\mathcal{J}^{k}(i,j)$, we use the same trick as in Ref.~\cite{Gong14}.
First, we consider the sum over $z_1$:
\begin{align}
    \sum_{z_1} J_{iz_1} J_{z_1z_2} \leq 2 \sum_{z_1:r_{iz_1}\leq r_{z_1j}}J_{iz_1} J_{z_1z_2},
\end{align}
where the right hand side sums only over $z_1$ being closer to $i$ than to $z_2$ and the factor $2$ accounts for exchanging the roles of $i$ and $z_2$.
We further separate the sum over $z_1$ in \cref{eq:Jconvo} into two, corresponding to whether $z_1$ is within a unit distance from $i$ or not:
\begin{align}
    \sum_{z_1} J_{iz_1} J_{z_1z_2} \leq 2\left( 
      \sum_{z_1:r_{iz_1}\leq 1} +\sum_{z_1:r_{iz_1}\geq 2} \right) J_{iz_1} J_{z_1z_2}   .\label{eq:thetwosums}
\end{align}
Since $r_{iz_1}\leq r_{z_1z_2}$, it follows that $r_{z_1z_2}\geq r_{iz_2}/2$.
Therefore, $J_{z_1z_2}\leq 2^\alpha J_{iz_2}$ and we further bound the second sum in \cref{eq:thetwosums} by
\begin{align}
    \sum_{z_1:r_{iz_1}\geq 2} J_{iz_1} J_{z_1z_2} 
    &\leq 2^\alpha J_{iz_2} \sum_{z_1:r_{iz_1}\geq 2} J_{iz_1}   \nonumber\\
    &\leq 2^\alpha J_{iz_2} 2^{1-\alpha} (\lambda-1)  \nonumber\\
    &\leq 2 (\lambda-1) \sum_{z_1:r_{iz_1}\leq 1} J_{iz_1} J_{z_1z_2},
\end{align}
where we bound $\sum_{z_1:r_{iz_1}\geq 2} J_{iz_1}\leq 2^{1-\alpha} (\lambda-1)$ and $J_{iz_2}\leq \sum_{z_1:r_{iz_1}\leq 1} J_{iz_1} J_{z_1z_2}$ similarly to Ref.~\cite{Gong14}.
Therefore, we have $\sum_{z_1} J_{iz_1} J_{z_1z_2}\leq 4\lambda\sum_{z_1:r_{iz_1}\leq 1} J_{iz_1} J_{z_1z_2}.$
Repeating this analysis for $z_2,\dots,z_k$ in \cref{eq:Jconvo}, we have an upper bound on $\mathcal{J}^{k}(i,j)$:
\begin{align}
    &\mathcal{J}^{k}(i,j) 
    \leq (4\lambda)^{k-1} \sum_{z_{1}:r_{iz_1}\leq 1}\sum_{z_{2}:r_{z_1z_2}\leq 1}\dots\nonumber\\
    &\quad\quad\quad\dots\sum_{z_{k-1}:r_{z_{k-2}z_{k-1}}\leq 1}
J_{iz_1}J_{z_{1}z_{2}}\dots  J_{z_{k-1}j}\label{eq:nnhop}\\
    &\leq (12\lambda)^{k-1} \times \begin{cases}
    1/{(r_{ij}-k+1)^\alpha}   & \text{ if $k<\mu r_{ij}$},\\
    1                               & \text{ if $k\geq \mu r_{ij}$},                      
    \end{cases}\\
    &\leq (12\lambda)^{k-1} \times \begin{cases}
    1/{[(1-\mu)r_{ij}]^\alpha}   & \text{ if $k<\mu r_{ij}$},\\
    1                               & \text{ if $k\geq \mu r_{ij}$},                      
    \end{cases}
\end{align}
where $\mu\in (0,1)$ is an arbitrary constant.

To get the second to last bound, we note that the maximum value that the summand in \cref{eq:nnhop} may achieve is $1/{(r_{ij}-k+1)^\alpha}$ when $k<\mu r_{ij}$ and 1 when $k\geq r_{ij}$, and the number of sites within a unit distance of any site is $3$.
Plugging this bound into \cref{eq:HK0} and \cref{eq:HK}, we have the Lieb-Robinson bound in Ref.~\cite{Gong14} generalized to many-body interactions:
\begin{align}
    &\norm{\comm{A(t),B}}\nonumber\\
    &\leq \norm A\norm B \sum_{i\in X}\sum_{j\in Y} \Bigg(\sum_{k=1}^{\ceil{\mu r_{ij}}-1}\frac{(24\lambda^2t)^k}{6\lambda k! [(1-\mu)r_{ij}]^\alpha} \nonumber\\
    &\quad\quad\quad\quad\quad\quad\quad\quad\quad\quad\quad\quad\quad+ \sum_{k=\ceil{\mu r_{ij}}}^\infty \frac{(24\lambda^2 t)^k}{6\lambda k!} \Bigg)\nonumber\\
    &\leq \norm A\norm B \sum_{i\in X}\sum_{j\in Y} Ce^{v t}\left(\frac{1}{ [(1-\mu)r_{ij}]^\alpha} + e^{-\mu r_{ij}} \right)\label{eq:LR-ZX-before-sum}\\
    &\leq \norm A\norm B \abs{X}\abs{Y} Ce^{v t}\left(\frac{1}{ [(1-\mu)r]^\alpha} + e^{-\mu r} \right),\label{eq:LR-ZX2}
\end{align}
where $C = 1/6\lambda$, $v = 24\lambda^2 $, and $r$ is, again, the distance between $X,Y$.
The proof for $D>1$ follows a very similar analysis. 
\end{proof}

A feature of \cref{eq:LR-ZX2} is that it depends on $\abs{X},\abs{Y}$, which can become problematic when $A,B$ are supported on a large number of sites. 
In such cases, we can sum over the sites of $X,Y$ in \cref{eq:LR-ZX-before-sum} to get more useful bounds.
Without any other assumptions, summing over the sites of $Y$ gives an extra factor of $r^{D}$:
{\medmuskip=-0mu
\thinmuskip=-0mu
\thickmuskip=-0mu
\nulldelimiterspace=0pt
\scriptspace=0pt
\begin{align}
    \norm{\comm{A(t),B}}\leq C \norm A \norm B \abs{X} 
    \left(
    \frac{1}{(1-\mu)^\alpha}
    \frac{e^{vt}}{r^{\alpha-D}}
    +e^{vt- \mu r}
    \right)
    ,\label{eq:LR-ZX-noY}
\end{align}}
where the constant $C$ absorbs all constants that may depend on $\mu$.
Note that the bound still depends on $\abs{X}$ but not on $\abs{Y}$.

We can go one step further and sum over the sites of $X$, but we need to assume that $X$ is convex (similarly to Ref.~\cite{Tran18}).
Then, we have a bound
{\medmuskip=-0mu
\thinmuskip=-0mu
\thickmuskip=-0mu
\nulldelimiterspace=0pt
\scriptspace=0pt
\begin{align}
    \norm{\comm{A(t),B}}\leq C \norm A \norm B \phi(X) 
    \left(
    \frac{1}{(1-\mu)^\alpha}
    \frac{e^{vt}}{r^{\alpha-D-1}}
    +e^{vt-\mu r}
    \right)
    ,\label{eq:LR-ZX-noY-noX}
\end{align}}
which is independent of $\abs{X}$. 
Here $\phi(X)$ is the boundary area of $X$, defined as the number of sites in $X$ that are adjacent to a site outside $X$.
\section{Absorption rate from linear response theory}\label{sec:app-linear-response}

This section provides more details on the derivation of the absorption rate within linear response theory.
In particular, we provide more mathematically rigorous proofs of \cref{EQN:LR_Exp_bound}~[\cref{sec:sigmaij_bound_proof}] and \cref{eq:linear-response-bound}~[\cref{sec:linear-response-bound-proof}].

\subsection{Proof of \cref{EQN:LR_Exp_bound}}
\label{sec:sigmaij_bound_proof}
In this section, we prove the statement of \cref{EQN:LR_Exp_bound} [also \cref{eq:sigmaij_bound} below].
We recall that the system Hamiltonian $H_0$ is a power-law Hamiltonian, while the harmonic drive $V(t)=g\cos(\omega t) O$ is a sum of local terms, $g\cos(\omega t)O_i$, each of which is supported on the site $i$ only, where $i$ runs over the sites of the system.
In addition, we assume that the system is initially in the equilibrium state $\rho_\beta$ of $H_0$ corresponding to the temperature $1/\beta$.
To the lowest order in $g$, the energy absorption rate of the system is proportional to the dissipative (imaginary) part of the response function, $\sigma(\omega) = \sum_{i,j} \sigma_{ij}(\omega)$, where $i,j$ are the sites of the system and 
\begin{align}
	\sigma_{ij}(\omega) = \frac{1}{2} \int_{-\infty}^{\infty} dt e^{i\omega t} \avg{\comm{O_i(t),O_j(0)}}_\beta,
\end{align}
where $\avg{X}_\beta \equiv \Tr(\rho_\beta X)$ denotes the expectation value of an operator $X$ in $\rho_\beta$.

In Ref.~\cite{Abanin15}, the authors proved that there exists constants $C,\kappa$ such that for all $\omega>0,\delta\omega>0$ and for all pairs $i,j$,  
\begin{align} 
	\abs{\sigma_{ij}([\omega,\omega+\delta\omega])}
	\leq C e^{-\kappa \omega}.\label{eq:sigmaij_bound}
\end{align}
The statement in Ref.~\cite{Abanin15} is for finite-range interactions, but, for completeness, we show here that it also holds for power-law Hamiltonians.
First, we consider the diagonal terms $\sigma_{ii}(\omega)$. 
Let $\ket{n}$ and $E_n$ denote the eigenstates and eigenvalues of $H_0$. 
Similarly to Ref.~\cite{Abanin15}, we rewrite $\sigma_{ii}(\omega)$ as
\begin{align}
	\sigma_{ii}(\omega) = \pi \sum_{n} p_n [\gamma_{ii}^{(n)}(\omega) - \gamma_{ii}^{(n)}(-\omega)],\label{eq:sigmaii} 
\end{align}
where $p_n$ is the probability that the state is in the eigenstate $\ket{n}$, and $\gamma_{ii}^{(n)}$ denotes the contribution to $\sigma_{ii}$ from the $n$-th eigenstate:
\begin{align}
	\gamma_{ii}^{(n)}(\omega) &= \sum_m \abs{\bra m O_i \ket{n}}^2 \delta(E_n-E_m-\omega) \nonumber\\
	&= \sum_{m}\frac{\abs{\bra{m}\ad_H^{k} O_i\ket{n}}^2}{\omega^{2k}}\delta(E_n-E_m-\omega),\label{eq:gammaii}
\end{align}
where $\ad_H O_i = \comm{H,O_i}$, $k$ is an integer to be chosen later, and the last equality comes from the fact that $\ket{m},\ket{n}$ are eigenstates of $H$ and the $\delta$ function fixes the energy difference to be $\omega$.

In Ref.~\cite{Abanin15}, the authors used the fact that $H$ has a finite range to upper bound the norm of $\ad_H^k O_i$ by $\lambda^k k!$ for some constant $\lambda$.
For power-law interactions, the proof does not apply because the Hamiltonian $H$ can contain interaction terms between arbitrarily far sites. 
Instead, we upper bound $\ad_H^k O_i$ by realizing that $O_i$ technically satisfies \cref{DEF_power-law} and is therefore a power-law Hamiltonian. 
It then follows from \cref{LEM_adHpl} in \cref{sec:math} that $\ad_H^k O_i \in \lambda^k k! \Hpl$, i.e. $\ad_H^k O_i$ is a power-law Hamiltonian up to a factor $\lambda^k k!$, where $\lambda$ is the same constant as in \cref{LEM_adHpl} and $\Hpl$ is the set of power-law Hamiltonians with exponent $\alpha$ [See \cref{sec:comm}].
Finally, we can upper bound
\begin{align}
	\norm{\ad_H^k O_i}\leq C \lambda^k k!, \label{eq:norm_ad_O}
\end{align}
by realizing that the supports of the terms in $\ad_H^k O_i$ all contain the site $i$.

Integrating \cref{eq:gammaii} over $\omega$, assuming $\delta\omega$ is small enough so that the number of energy levels in the range $[\omega,\omega+\delta\omega]$ is finite, and using \cref{eq:norm_ad_O}, we have
\begin{align}
	\abs{\gamma_{ii}^{(n)}([\omega,\omega+\delta\omega])} 
	&\leq C\left(\frac{\lambda^k k!}{\omega^k}\right)^2 \nonumber\\
	&\leq C\left(\frac{\lambda k}{\omega}\right)^{2k} 
	\leq C e^{-\kappa \omega},
\end{align}
where $\kappa = 2/(\lambda e)$ and, to get the last line, we choose $k = \omega/(\lambda e)$.
Plugging this bound into \cref{eq:sigmaii} and summing over $n$ yields \cref{EQN:LR_Exp_bound} for $i = j$.
The bound for $i\neq j$ can be derived using the positivity of $\sigma$~\cite{Abanin15} and the Cauchy-Schwartz inequality,
\begin{align}
	\abs{\sigma_{ij}(\omega)} \leq \frac{1}{2} [\sigma_{ii}(\omega)+\sigma_{jj}(\omega)].
\end{align}
Therefore, \cref{EQN:LR_Exp_bound} applies for all power-law Hamiltonians $H$.

\subsection{Proof of \cref{eq:linear-response-bound}} \label{sec:linear-response-bound-proof}
We now provide a rigorous proof of \cref{eq:linear-response-bound} in the main text. 
\Cref{eq:sigmaij_bound} says that the $(i,j)$ entry of $\sigma([\omega,\omega+\delta\omega])$ is exponentially suppressed.
In principle, summing over all $i,j$ implies that $\sigma([\omega,\omega+\delta\omega])$ is also exponentially small as a function of $\omega$.
However, since there are $N$ sites in the system, this summation results in an additional factor of $N^2$, making $\sigma([\omega,\omega+\delta\omega])$ superextensive.
Therefore, this naive calculation breaks down in the thermodynamic limit ($N\rightarrow \infty$).

Instead, to show that $\sigma([\omega,\omega+\delta\omega])$ increases only as fast as $ N$, we use  Lieb-Robinson bounds to bound the off-diagonal terms $\sigma_{ij}(\omega)$.
Let $r_{ij} = \dist(i,j)$ denote the distance between the pair of sites $i,j$.
Without loss of generality, we assume $\omega\geq 2\delta\omega$.
We can then bound
\begin{align}
	&\sigma([\omega,\omega+\delta\omega])
	=\int_\omega^{\omega+\delta\omega} d\omega' \sigma(\omega')\nonumber\\
	&\leq c_1 \int_{-\infty}^{\infty} d\omega' e^{-(\frac{\omega'-\omega}{\delta\omega})^2}\sigma(\omega')\nonumber\\
	&=c_2\delta\omega \sum_{i,j}\int_{-\infty}^{\infty} dt e^{(-t/\delta t)^2} e^{-i\omega t} \avg{\comm{O_i(t),O_j}},\label{eq:to-time-domain}
\end{align}
where $c_1 = \frac{e}{1-e^{-8}}$, $c_2 = c_1\sqrt{\pi}/2$, which we will combine and denote by $C$, and $\delta t=2/\delta\omega$. 
The first inequality is because $\sigma(\omega)$ is positive for $\omega>0$ and $\sigma(-\omega) = -\sigma(\omega)$.
The second equality comes from evaluating the integral over $\omega'$.
We then use the Lieb-Robinson bound in Ref.~\cite{Gong14}, which applies for interactions with characteristic exponent $\alpha > D$:
\begin{equation}\label{EQN:Ctr_Gong}
	\norm{[O_{i}(t), O_{j}(0)]}\leq C e^{vt}\left(\frac{1}{r_{ij}^{\alpha}}+e^{-\mu r_{ij}}\right),
\end{equation}
where $v,C,\mu$ are positive constants.
While this bound was derived in Ref.~\cite{Gong14} for 2-body interactions, it also holds for more general $k$-body interactions and thus for fully general power-law Hamiltonians [see \cref{eq:LR-ZX2}]. 

We now divide the sum in \cref{eq:to-time-domain} into two parts corresponding to $r_{ij}>r_*$ and $r_{ij}\leq r_*$ for some parameter $r_*$ we shall choose later.
The sum over $i,j$ such that  $r_{ij}>r_{*}$ can then be bounded by first inserting \cref{EQN:Ctr_Gong} into \cref{eq:to-time-domain} and evaluating the integration over time. 
Note that the factor $ e^{-t^2/\delta t^2}$ suppresses the contribution from $e^{vt}$ at large $t$. 
Therefore, performing the integral yields an upper bound $C(1/r_{ij}^{\alpha}+e^{-\mu r_{ij}})$ for each term corresponding to the pair $(i,j)$, and the sum over $r_{ij}>r_*$ gives:
\begin{align}
\sum_{i,j:r_{ij} >r_{*}}C\left(\frac{1}{r_{ij}^{\alpha}}+e^{-\mu r_{ij}}\right) 
\leq  CN \left( \frac{1}{r_{*}^{\alpha-D}}+e^{-\mu r_{*}}\right),\label{eq:sum>r_*}
\end{align} 
for $\alpha >D$, where the factor of $N$ comes from summing over $i$ and the factor of $r^{D}$ comes from summing over $j$.

On the other hand, for $r_{ij}\leq r_*$, we simply use \cref{eq:sigmaij_bound} to bound their contributions. 
Summing over $i,j$ such that $r_{ij}\leq r_*$, we get a bound $CN r_*^D e^{-\kappa \omega}$, where the factor of $N$ again comes from summing over $i$ and the factor of $r_*^D$ from counting the number of sites $j$ within a distance $r_*$ from $i$. 
Combining with \cref{eq:sum>r_*} yields an upper bound on the the total heating rate
\begin{align}
	|\sigma([\omega, \omega + \delta\omega])| \leq CN r_*^D \left(e^{-\kappa \omega}+\frac{1}{r_*^\alpha}+r_*^{-D}e^{-\mu r_*}\right).
\end{align}
Choosing $r_{*}^\alpha = e^{\kappa\omega}$ and noting that the last term is dominated by the first two when $\omega$ is large enough, we find
\begin{align}
	|\sigma([\omega, \omega + \delta\omega])| \leq CN e^{-\frac{\alpha-D}{\alpha}\kappa\omega},
\end{align}
which is exponentially small with $\omega$ as long as $\alpha>D$.  
\section{The effective Hamiltonian}\label{sec:Gq}
In this section, we study the structure of the effective Hamiltonian defined in \cref{eq:Hprime}.
Specifically, we show that the operators $G_q$ defined in \cref{eq:Gq} are also power-law Hamiltonians [See also \cref{LEM_Gq} in the main text for $q<\qmax$ and \cref{LEM_Gq>} below for $q\geq\qmax$].
In addition, we show that the norm $G_q$ for $q\geq \qmax$ is exponentially small as a function of $q$ and $\nmax$ [\cref{LEM_Gq>}], implying that the norm of the residual drive $V'$ is also exponentially small. 

\subsection{Structure of $G_q$ for $q<\qmax$}
\label{app:LemGq>proof}
First, we prove the statement of \cref{LEM_Gq} that the operators $G_q$ are also power-law Hamiltonians for all $q<\qmax$. 
\begin{proof}
We proceed by induction and assume that \cref{LEM_Gq} holds for all $q$ up to $q=q_0-1$ for some $q_0\geq 1$.
We now prove that it also holds for $q=q_0$.
We consider the first term in the definition of $G_{q_0}$ [\cref{eq:Gq}]:
\begin{align}
    G_{{q_0},1} = \sum_{k=1}^{q_0} \frac{(-1)^k}{k!} \sum_{\substack{
            1\leq i_1,\dots,i_k\leq q_0\\
            i_1+\dots+i_k=q_0}} 
    \ad_{\Omega_{i_1}}\dots\ad_{\Omega_{i_k}} H(t).
\end{align}
\begin{widetext}
Using \cref{LEM_Gq} (note that it applies to all $i\leq q_0$) and \cref{LEM_adHpl} in \cref{sec:math}, we have
\begin{align}
    G_{{q_0},1} 
    &\in \sum_{k=1}^{q_0}\frac{1}{k!}
    \sum_{\substack{
            1\leq i_1,\dots,i_k\leq q_0\\
            i_1+\dots+i_k=q_0}} 
    T^{q_0} c^{q_0} \lambda^{q_0-k}
    \prod_{j=1}^k (i_j-1)! q_0^kc^{-k}\lambda^k\Hpl^{(q_0+1)}\nonumber\\
    &=T^{q_0} c^{q_0}\lambda^{q_0}\sum_{k=1}^{q_0}\frac{q_0^kc^{-k}}{k!}
    \sum_{\substack{
            1\leq i_1,\dots,i_k\leq q_0\\
            i_1+\dots+i_k=q_0}} \prod_{j=1}^k (i_j-1)! \Hpl^{(q_0+1)}\nonumber\\
    &\subseteq T^{q_0} c^{q_0}\lambda^{q_0}\sum_{k=1}^{q_0}\frac{q_0^kc^{-k}}{k!}(q_0-k)! 2^k \Hpl^{(q_0+1)}
    \nonumber\\    
    &
    \subseteq T^{q_0} c^{q_0}\lambda^{q_0}q_0!\underbrace{\sum_{k=1}^{q_0}\frac{2^kq_0^kc^{-k}(q_0-k)! }{q_0!k!}}_{\leq c_1} \Hpl^{(q_0+1)} \nonumber\\
    &\subseteq c_1T^{q_0} c^{q_0}\lambda^{q_0}q_0! \Hpl^{(q_0+1)},\label{EQ_Gq1<}
\end{align}
where $c_1$ is a constant which exists because the sum over $k$ converges [See \cref{lem:c1} in \cref{sec:math}]. 
\end{widetext}
To get the first equation, we use \cref{LEM_adHpl}, with $k_{\max}$ upper bounded by $q_0$ every time. 
We have also used the second part of \cref{LEM_Sum_inv_mult} in the Appendix to bound the sum over $i_1,\dots,i_k$.

Next, we consider the second term in the definition of $G_{q_0}$:
{
\medmuskip=-0mu
\thinmuskip=-0mu
\thickmuskip=-0mu
\nulldelimiterspace=0pt
\scriptspace=0pt
\begin{align}
     G_{q_0,2} 
    &=  i\sum_{k=1}^{q} \frac{(-1)^{k+1}}{(k+1)!}
    \sum_{\substack{
            1\leq i_1,\dots,i_k,m\leq q+1\\
            i_1+\dots+i_k+m=q+1}}
    \ad_{\Omega_{i_1}}\dots\ad_{\Omega_{i_k}} \partial_t \Omega_m(t).
\end{align}
}
\begin{widetext}
Again, we use \cref{LEM_Gq} and \cref{LEM_adHpl} to show that
\begin{align}
    G_{q_0,2} 
    &\in 
    \sum_{k=1}^{{q_0}} \frac{q_0^k}{(k+1)!}
    \sum_{\substack{
            1\leq i_1,\dots,i_k,m\leq {q_0}+1\\
            i_1+\dots+i_k+m={q_0}+1}}
        T^{{q_0}}c^{q_0-k-1}\lambda^{{q_0}-k-1}
        \prod_{j=1}^k(i_j-1)! (m-1)!   
        \lambda^k \Hpl^{(q_0+1)}\\
        &= T^{q_0} c^{q_0 }\lambda^{q_0}  
    \sum_{k=1}^{{q_0}} \frac{q_0^kc^{-k}}{(k+1)!}
    \underbrace{\sum_{\substack{
                1\leq i_1,\dots,i_k,m\leq {q_0}+1\\
                i_1+\dots+i_k+m={q_0}+1}}}_{\leq 2^{k+1}}
    \underbrace{\prod_{j=1}^k(i_j-1)! (m-1)!}_{\leq (q_0+1-(k+1))! = (q_0-k)!} \Hpl^{(q_0+1)}\\
    &\subseteq  T^{q_0} c^{q_0 }\lambda^{q_0}  
        2\sum_{k=1}^{{q_0}} \frac{2^kq_0^kc^{-k}}{(k+1)\!}(q_0-k)! 
        \Hpl^{(q_0+1)}\\
    &\subseteq  2T^{q_0} c^{q_0 }\lambda^{q_0}q_0!  
        \underbrace{\sum_{k=1}^{q_0} \frac{2^kq_0^kc^{-k}}{k!}\frac{(q_0-k)!}{q_0!}}_{\leq c_1} \Hpl^{(q_0+1)}
    \subseteq 2c_1  T^{q_0} c^{q_0 }\lambda^{q_0}  q_0!  \Hpl^{(q_0+1)},\label{EQ_Gq2<}
\end{align}
where we have used \cref{LEM_Sum_inv_mult} in \cref{sec:math} to bound the sums over $i_1,\dots,i_k, m$.
\end{widetext} 
Combining \cref{EQ_Gq1<} and \cref{EQ_Gq2<}, we have
\begin{align}
    G_{q_0}\in 3c_1 T^{q_0} c^{q_0}q_0! \lambda^{q_0} \Hpl^{(q_0+1)}.
\end{align}
Note that $c_1$ can be made arbitrarily small by choosing a larger value for $c$.
Therefore, with $c$ large enough so that $3c_1<1$, we have that \cref{LEM_Gq} holds for $q = q_0$.
\end{proof}\subsection{Structure of $G_q$ for $q\geq \qmax$}
\label{sec:proofGq>}

We now prove \cref{eq:Gq>expsmall}, which is a similar result to \cref{LEM_Gq}, but for $q\geq \qmax = \nmax$.
\begin{lemma}\label{LEM_Gq>}
For all $q\geq \qmax = \nmax$, 
$
    G_q \in C e^{-\kappa' q} \Hpl,
$
where $C$ and $\kappa'$ are constants.
\end{lemma}
\begin{proof}
    Let us first look at the first term in \cref{eq:Gq}:
\begin{align}
G_{q,1} = \sum_{k=1}^q \frac{(-1)^k}{k!} \sum_{\substack{
        1\leq i_1,\dots,i_k\leq \nmax\\
        i_1+\dots+i_k=q}} 
\ad_{\Omega_{i_1}}\dots\ad_{\Omega_{i_k}} H(t).
\end{align} 
We also recall from \cref{LEM_Gq} that for all $q\leq \nmax$, 
\begin{align}
\Omega_{q} \in T^q (q-1)! c^{q-1} \lambda^{q-1} \Hpl^{(q)} 
\subseteq \frac{1}{\lambda cq}T^q q! c^q \lambda^{q} \Hpl.
\end{align}
For all $q\leq \nmax$, we have
\begin{align}
    T^q q! c^q \lambda^q \leq (Tc\lambda q)^q \leq (Tc\lambda \nmax)^q\leq e^{-\kappa q},
\end{align}
where we have used $\nmax = e^{-\kappa} /{(Tc\lambda)}$.
Therefore, for all $q\leq \nmax$, we have
\begin{align}
\Omega_{q} \in \frac{1}{\lambda c q} T^q q! c^q \lambda^{q} \Hpl \in \frac{1}{\lambda c q}e^{-\kappa q}\Hpl.
\end{align}
Note also that $H(t)\in \Hpl$. 
Therefore, using \cref{LEM_adHpl}, we have 
\begin{align}
    \ad_{\Omega_{i_1}}\dots\ad_{\Omega_{i_k}} H(t) 
    &\in\frac{1}{i_1\dots i_k} \frac{q^k}{c^k} e^{-\kappa q}\lambda^{-k}\lambda^{k} \Hpl \nonumber\\
    &= \frac{1}{i_1\dots i_k} \frac{q^k}{c^k}e^{-\kappa q} \Hpl.
\end{align}
Thus, we get for all $q$:
\begin{align}
    G_{q,1} &\in \underbrace{\left(\sum_{k=1}^q \frac{q^k}{c^k k!} \sum_{\substack{
            1\leq i_1,\dots,i_k\leq \nmax\\
            i_1+\dots+i_k=q}} \frac{1}{i_1\dots i_k}\right)}_{\leq e^{q/c} 2^q} e^{-\kappa q}\Hpl\nonumber\\ 
            &\subseteq  e^{-(\kappa-\ln 2-1/c)q} \Hpl.\label{EQ_Gq1}
\end{align}
Note that the $i_{j} \leq \nmax$ as we only define $\Omega$ up to $\nmax$. Further, the factor of $2^{q}$ comes from upper-bounding $\frac{1}{i_{1}\dots i_{k}}$ with 1 and the number of terms with $2^{q}$. Next, we consider the second term in the definition of $G_q$:
{\medmuskip=-0mu
\thinmuskip=-0mu
\thickmuskip=-0mu
\nulldelimiterspace=0pt
\scriptspace=0pt
\begin{align}
    G_{q,2} = i \sum_{k=1}^{q} \frac{(-1)^{k+1}}{(k+1)!}\sum_{\substack{
            1\leq i_1,\dots,i_k,m\leq q+1\\
            i_1+\dots+i_k+m=q+1}}
    \ad_{\Omega_{i_1}}\dots\ad_{\Omega_{i_k}} \partial_t \Omega_m(t). 
\end{align}}
Note that
\begin{align}
    \partial_t \Omega_m(t) \in T^{m-1} (m-1)! c^{m-1} \lambda^{m-1} \Hpl 
    \subseteq e^{-\kappa(m-1)} \Hpl.
\end{align}
Thus, we have
\begin{align}
    G_{q,2}&\in     \left(\sum_{k=1}^{q} \frac{q^{k}}{c^k(k+1)!}\sum_{\substack{
            1\leq i_1,\dots,i_k,m\leq q+1\\
            i_1+\dots+i_k+m=q+1}}  e^{-\kappa q}\right)\Hpl\nonumber\\
        &\subseteq      2e^{-(\kappa-\ln 2-1/c)q} \Hpl.\label{EQ_Gq2}
\end{align}
Combining \cref{EQ_Gq1} and \cref{EQ_Gq2}, we arrive at \cref{LEM_Gq>} with $\kappa' = \kappa-\ln2 - 1/c$, which can be made to be positive by choosing $\kappa>\ln 2 + 1/c$. It suffices, however, to choose $\kappa > \ln 2$, since making $c$ large enough sends $1/c$ to zero. 
\Cref{eq:V'} also follows.
\end{proof}\section{Using Lieb-Robinson bounds for evolutions of local observables}
\label{sec:applylr}
In this section, we use the Lieb-Robinson bounds to bound the norm of $\delta$ in \cref{eq:deltanorm}.
In the main text, we argue that $\norm{\delta(t)}$ would be small up to time $t_*\propto \omega_*$ if the light cone induced by the Lieb-Robinson bound is logarithmic, and $t_*\propto e^{\kappa'\omega_*}$ if the light cone is algebraic.
We provide below the mathematical details to supplement the argument.

Recall that $V'(t)\in C e^{-\kappa' \nmax} \Hpl, H_* \in \gamma \Hpl$ ($\gamma$ is a constant that depends only on $\kappa,\alpha$) and that we defined the normalized $V'' = C^{-1} e^{\kappa' \nmax} V, \bar H'' = \gamma^{-1} H_*$ such that:
\begin{align}
\norm{\delta}\leq C e^{-\kappa' \nmax}\int_0^t ds \norm{\comm{V''(s),e^{is\gamma \bar H''}Oe^{-is\gamma\bar H''}}}.
\end{align}
We now use a Lieb-Robinson bound for power-law interactions to bound the commutator. 
The idea is that for a finite time $s$, the operator $O$ mostly spreads within a light cone, and only the terms of $V''(s)$ within the light cone significantly contribute to the commutator. 

In contrast to the finite-range interacting Hamiltonians,
a tight Lieb-Robinson bound has yet to be proven for power-law Hamiltonians with finite $\alpha>D$.
In the following sections, we consider the effect of using different Lieb-Robinson bounds, namely the bounds in Gong~\etal~\cite{Gong14}, Else~\etal~\cite{Else18}, Tran~\etal~\cite{Tran18}. The case of a hypothetical bound, which would be tight if it were proven, is treated in the main text.

\subsection{Using Gong~\etal~\cite{Gong14}'s bound}
First, we consider a generalization of the bound in Gong~\etal~\cite{Gong14} [See also \cref{eq:LR-ZX-noY}].
The bound holds for $\alpha>D$, has a logarithmic light cone $t\gtrsim \log r$, and is extended to many-body interactions. 
To bound the commutator norm in \cref{EQ_delta_int}, recall that we write $V''(s) = \sum_{r=0}^\infty V''_r(s)$, where $V''_r(s)\equiv \sum_{X:\dist(X,O)\in[r,r+1)} h_X$ denotes the terms of $V''(s)$ supported on subsets exactly a distance between $r$ and $r+1$ away from $O$.
furthermore, since $V''(s)$ is a power-law Hamiltonian, it follows that $\norm{V''_r(t)}\leq C r^{D-1}$.

From \cref{eq:LR-ZX-noY}, the light cone of the bound is $r_*(s)= e^{vs/\alpha}$.
We further divide $V''_r(s)$ into those with $r\leq r_*(s)$ and $r>r_*(s)$.
In the former case, we simply bound:
\begin{align}
        &\sum_{r\leq r_*(s)} \norm{\comm{V''_r(s),e^{is\gamma \bar H''}Oe^{-is\gamma\bar H''}}}\nonumber\\
        &\leq 2\sum_{r\leq r_*(s)}  \norm{V''_r(s)}\norm {O} \leq C r_*(s)^{D}\leq C e^{Dvs/\alpha}.\label{eq:usingZX-in}
\end{align}

For the latter case, we use \cref{eq:LR-ZX-noY} to bound the commutator norm:
\begin{align}
        &\sum_{r>r_*(s)} \norm{\comm{V''_r(s),e^{is\gamma \bar H''}Oe^{-is\gamma\bar H''}}}\nonumber\\
        &\leq C\sum_{r> r_*(s)}  \norm{V''_r(s)}\norm {O} \left(\frac{e^{vs}}{r^{\alpha-D}}+e^{vs-\mu r}\right)\\
        &\leq C\sum_{r> r_*(s)}  \left(\frac{e^{vs}}{r^{\alpha-2D+1}}+r^{D-1}e^{vs-\mu r}\right)\\
        &\leq C\left(\frac{e^{vs}}{r_*(s)^{\alpha-2D}}+r_*(s)^{D-1}e^{vs-\mu r_*(s)}\right)\\
        &\leq C\left(e^{2Dvs/\alpha}+e^{vs\frac{D-1}{\alpha}}e^{vs-\mu e^{vs/\alpha}}\right)\\
        &\leq C e^{2Dvs/\alpha},\label{eq:usingZX-out}
\end{align}
where we use the same $C$ to denote different constants that may depend on $\mu,\alpha$. 
Note that while the bound in \cref{eq:LR-ZX-noY} is valid for $\alpha>D$, the sum over $r$ converges only when $\alpha>2D$.

Plugging \cref{eq:usingZX-in} and \cref{eq:usingZX-out} into \cref{EQ_delta_int} and integrating over $s$, we have
\begin{align}
        \norm{\delta}\leq C e^{-\kappa' \nmax} e^{2Dvt/\alpha},
\end{align}
which is the result presented in \cref{sec:PRB}. 
Again, $\delta$ is only small for up to time $t_*\propto \nmax\propto 1/T$, which is expected because the region inside the light cone implied by this bound expands exponentially fast with time. 

\subsection{Using Else~\etal~\cite{Else18}'s bound}
Instead of using Gong~\etal~'s bound, we now use the bound in Else~\etal~\cite{Else18}, which already holds for many-body interactions.
The bound states that when $\abs{X}=1$,
{\medmuskip=-0mu
\thinmuskip=-0mu
\thickmuskip=-0mu
\nulldelimiterspace=0pt
\scriptspace=0pt
\begin{align}
    \norm{\comm{A(t),B}}\leq C \norm A\norm B \left\{
     \exp\left(vt-r^{1-\sigma}\right)
     +\frac{(vt)^{1+D/(1-\sigma)}}{r^{\sigma (\alpha-D)}}
    \right\},\label{eq:LR-Else}
\end{align}}
where $1>\sigma>(D+1)/(\alpha-D+1)$ is a constant that we can choose. 
Since our aim is to prove an exponential heating time for $\alpha$ as small as possible, we need the algebraic tail exponent $\sigma(\alpha-D)$ to be as large as possible.
So we will assume that we pick some $\sigma$ very close to 1.

First, let us look at the light cone generated by \cref{eq:LR-Else}.
The first term of the bound gives a light cone $t\gtrsim r^{1-\sigma}$, while the second term gives $t\gtrsim r^{(1-\sigma)\frac{\sigma(\alpha-D)}{D+1-\sigma}}$.
Since we are choosing $\sigma$ close to 1, $\frac{\sigma(\alpha-D)}{D+1-\sigma}$ will be larger than 1 when $\alpha>2D$.
The former light cone, i.e. $t\gtrsim r^{1-\sigma}$, is therefore looser and thus dominates the latter. 
In the rest of the calculation, we take $r_*(t) = t^{1/(1-\sigma)}$ to be the light cone boundary.

Similar to \cref{eq:usingZX-in}, we get an upper bound for the terms inside the light cone:
\begin{align}
        &\sum_{r\leq r_*(s)} \norm{\comm{V''_r(s),e^{is\gamma \bar H''}Oe^{-is\gamma\bar H''}}}\nonumber\\
        &\quad\quad\quad\quad\leq 2\sum_{r\leq r_*(s)}  \norm{V''_r(s)}\norm {O} \nonumber\\
        &\quad\quad\quad\quad\leq C r_*(s)^{D}
        \leq C s^{D/(1-\sigma)}.\label{eq:usingElse-in}
\end{align}

For the terms outside the light cone, we use \cref{eq:LR-Else}:
\begin{align}
        &\sum_{r>r_*(s)} \norm{\comm{V''_r(s),e^{is\gamma \bar H''}Oe^{-is\gamma\bar H''}}}\nonumber\\
        &\leq \sum_{r> r_*(s)}  \norm{V''_r(s)}\norm {O} \left(e^{vs-r^{1-\sigma}}+\frac{(vs)^{1+D/(1-\sigma)}}{r^{\sigma (\alpha-D)}}\right)\nonumber\\
        &\leq C\sum_{r> r_*(s)}   \left(r^{D-1}e^{vs-r^{1-\sigma}}+\frac{(vs)^{1+D/(1-\sigma)}}{r^{\sigma (\alpha-D)-D+1}}\right)\nonumber\\
        &\leq C\left(\frac{1}{D}\xi\left(\frac{D}{1-\sigma}\right)e^{vs} {r_*^{D}e^{-r_*^{1-\sigma}}}+\frac{(vs)^{1+D/(1-\sigma)}}{r_*^{\sigma (\alpha-D)-D}}\right)\nonumber\\
        &\leq C\left(\xi\left(\frac{D}{1-\sigma}\right) s^{D/(1-\sigma)}+\frac{(vs)^{1+D/(1-\sigma)}}{s^{\frac{\sigma (\alpha-D)-D}{1-\sigma}}}\right)\nonumber\\
        &\leq C\xi\left(\frac{D}{1-\sigma}\right) s^{\frac{D}{1-\sigma}},
        \label{eq:usingElse-out}
\end{align}
where $\xi(x) \equiv \frac{1}{x}2^x \Gamma(x) $, $\Gamma$ is the Gamma function, and we again absorb all constants that may depend on $D$ alone into the constant $C$. 
We drop the second term in the second to last inequality because for $\sigma$ arbitrarily close to 1 and $\alpha > 2D$ (see below), the second term may be upper-bounded by the first.
To estimate the sum over $r$, we have used \cref{lem:forElse} in \cref{sec:sums}.
Plugging \cref{eq:usingElse-in} and \cref{eq:usingElse-out} into \cref{EQ_delta_int} and integrating over time, we get
\begin{align}
        \norm{\delta}\leq C e^{-\kappa' \nmax} \xi\left(\frac{D}{1-\sigma}\right)  t^{\frac{D}{1-\sigma}+1}.
\end{align}
Thus, the difference is small up to an exponentially long time $t_*\propto e^{\kappa' \nmax\frac{1-\sigma}{D + 1 - \sigma}}$.
The sum over $r$ converges if $\sigma(\alpha-D)>D$, or equivalently $\alpha > D\left(1+\frac{1}{\sigma}\right)$. 
Since $\sigma$ can be chosen arbitrarily close to 1, this condition is effectively equivalent to $\alpha>2D$.

One should be careful, however, in taking the limit $\sigma$ goes to one since 
i) the heating time $t_*\propto e^{\kappa' \nmax\frac{1-\sigma}{D + 1 - \sigma}}$ is no longer exponential in $\nmax$ and
ii) the prefactor ${\xi\left(\frac{D}{1-\sigma}\right)}$ diverges faster than exponentially in this limit. 
Nevertheless, the analysis is still valid for fixed values of $\sigma<1$. 

\subsection{Using Tran~\etal~\cite{Tran18}'s bound}

In addition to Else~\etal~\cite{Else18}'s bound, we can also use the bound in Tran~\etal~\cite{Tran18} [see also \cref{eq:Minh-r0=1} for a generalization to $k$-body interactions], which also works for $\alpha>2D$.
Compared to the bound in Else~\etal~, the bound in Tran~\etal~ has a tighter light cone $r_*(s) = s^{(\alpha-D)/(\alpha-2D)}$, but it decays with the distance $r$ as $r^{\alpha-2D}$, slower than the tail $r^{\sigma(\alpha-D)}$ in Else~\etal~ when $\sigma>(\alpha-2D)/(\alpha-D)$.

Similar to before, we further divide $V''_r(s)$ into those with $r\leq r_*(s)$ and $r>r_*(s)$.
For the terms inside the light cone, we again bound:
\begin{align}
        &\sum_{r\leq r_*(s)} \norm{\comm{V''_r(s),e^{is\gamma \bar H''}Oe^{-is\gamma\bar H''}}}\nonumber\\
        &\leq 2\sum_{r\leq r_*(s)}  \norm{V''_r(s)}\norm {O} \nonumber\\
        &\leq C r_*(s)^{D}
        \leq C s^{D(\alpha-D)/(\alpha-2D)}.\label{eq:usingMinh-in}
\end{align}

For the terms outside the light cone, we use \cref{eq:Minh-r0=1} with $\phi(X) = 1$:
\begin{align}
        &\sum_{r>r_*(s)} \norm{\comm{V''_r(s),e^{is\gamma \bar H''}Oe^{-is\gamma\bar H''}}}\nonumber\\
        &\leq \sum_{r> r_*(s)}  \norm{V''_r(s)}\norm {O} \left(\frac{s^{\alpha-D}}{r^{\alpha-2D}}+sr^{D-1}e^{-r/s}\right)
        \nonumber\\
        &\leq C\sum_{r> r_*(s)}  \left(\frac{s^{\alpha-D}}{r^{\alpha-3D+1}}+sr^{2D-2}e^{-\mu r/s}\right)\nonumber\\
        &\leq C\left(\frac{s^{\alpha-D}}{r_*(s)^{\alpha-3D}}+s^2r_*(s)^{2D-2}e^{-\mu r_*(s)/s}\right)\nonumber\\
        &\leq C\left(s^{\frac{D(\alpha-D)}{\alpha-2D}}+s^2s^{2(\alpha-D)(D-1)/(\alpha-2D)}e^{-\mu s^{D/(\alpha-2D)}}\right)\nonumber\\
        &\leq  Cs^{\frac{D(\alpha-D)}{\alpha-2D}},\label{eq:usingMinh-out}
\end{align}
where we have dropped the second term in the second to last inequality because it is exponentially small in $s$ and can be upper bounded by the first term.
Note that we require $\alpha>3D$ in order for the sum over $r$ to converge. 

Plugging \cref{eq:usingMinh-in} and \cref{eq:usingMinh-out} into \cref{EQ_delta_int} and integrating over time, we get
\begin{align}
        \norm{\delta}\leq C e^{-\kappa' \nmax} t^{\frac{D(\alpha-D)}{\alpha-2D}+1}.
\end{align}
Thus, the difference is small up to an exponentially long time $t_*\propto e^{\kappa' \nmax\frac{\alpha-2D}{\alpha(D+1)-D(D+2)}}$.
Compared to using Else~\etal~'s bound, this analysis works only when $\alpha>3D$.
However, within this regime, the exponent of the heating time using this bound is larger than using Else~\etal~.
This is a manifestation of the trade-off between the tail and the light cone when switching from Else~\etal~ to Tran~\etal~ bound. 

\section{Mathematical preliminaries}\label{sec:math}
This section contains mathematical details omitted from the previous sections for clarity.
In \cref{sec:comm}, we discuss the properties of the set of power-law Hamiltonians defined in \cref{DEF_power-law}.
In \cref{sec:sums}, we present some bounds on discrete sums.

\subsection{Properties of the set $\Hpl$ of power-law Hamiltonians}
\label{sec:comm}

In this section, we explore some properties of $\Hpl$ that are
useful for proving that the effective Hamiltonian is also power-law [See \cref{sec:Gq}].

We recall from the main text that $\Hpl$ is the set of power-law Hamiltonians with the exponent $\alpha$.
In addition, $\Hpl^{(k)}$ is the subset of $\Hpl$ which contains all power-law Hamiltonians whose local support size [see \cref{DEF_power-law}] is at most $k+1$.
For a real positive constant $a$, we also denote by $a\Hpl$ the set of Hamiltonians $H$ such that $a^{-1} H$ is a power-law Hamiltonian with the exponent $\alpha$. 
It is straightforward to prove the following identities:
\begin{align}
    &a \Hpl + b \Hpl \subset (a+b)\Hpl, \\
    &a \Hpl \subset b \Hpl \text{ if } a\leq b.
\end{align}
The following lemma is particularly useful for the adjoint operation:
\begin{lemma}\label{LEM_adHpl}
For $\alpha>D$, if $H_1 \in a \Hpl^{(k_{1})}, H_2\in b \Hpl^{(k_2)}$ for some positive constants $a,b,k_1,k_2$, then $\ad_{H_1} H_2 \in ab \lambda k_{\max}\Hpl^{(k_1+k_2)}$, where $\lambda$ is a constant to be defined later and $k_{\max}=\max\{k_1,k_2\}$.
\end{lemma}

\begin{proof}
Write  $H_1 = \sum_X a_X,H_2 = \sum_Y b_Y, \ad_{H_1} H_2 = \sum_Z h_Z$, where $h_Z = \ad_{h_X} h_Y$ and $Z = X\cup Y$. By our definition of power-law Hamiltonians, we have:
\begin{align}
    \sum_{X \ni i,j}\norm{a_{X}} \leq \frac{a}{\dist(i,j)^{\alpha}}, 
    \quad \sum_{Y \ni i,j}\norm{b_{Y}} \leq \frac{b}{\dist(i,j)^{\alpha}}. 
\end{align}
When $\alpha>D$, it is also straightforward to prove that $\sum_{X\ni i} \norm{a_X} \leq a \lambda_0$ for all $i$, where $\lambda_0$ is a constant that depends only on $\alpha,D$.

Note that $h_Z\neq 0$ only if $X\cap Y\neq \varnothing$. 
We seek to bound $\sum_{Z\ni i,j}\norm{h_Z}$ which sums over $Z = X\cup Y \ni i,j$ and $X \cap Y \neq \varnothing$. 
We discuss some useful notations. 
We will occasionally rewrite or label summations with restrictions using the indicator function $\xi(A)$ where $\xi(A) = 1,0$ if $A$ is true, false respectively. 
There are 9 mutually exclusive cases [\cref{TAB:LEM_adHPl_indicators}], satisfying $i, j \in X \cup Y$ depending on whether $i,j$ are in $X,Y$, or both.

\begin{table}[ht!]
\centering
\begin{tabular}{CCCCC}
\toprule
         & \in X & \notin X & \in Y & \notin Y \\ 
\colrule 
\xi_{1} & i,j   & -      & i,j   & -      \\ 
\xi_{2} & i,j   & -      & i     & j        \\ 
\xi_{3} & i,j   & -      & j     & i        \\ 
\xi_{4} & i,j   & -      & -   & i,j      \\ 
\xi_{5} & i     & j        & i,j   & -      \\ 
\xi_{6} & i     & j        & j     & i        \\ 
\xi_{7} & j     & i        & i,j   & -      \\ 
\xi_{8} & j     & i        & i     & j        \\ 
\xi_{9} & -   & i,j      & i,j   & -      \\ 
\botrule
\end{tabular}
\caption{Mutually exclusive indicator functions for Lemma \ref{LEM_adHpl}. For example, $\xi_1= 1$ if all of the conditions in the first row, i.e. $i,j\in X$ and $i,j\notin Y$, hold and $\xi_1 = 0$ otherwise.}\label{TAB:LEM_adHPl_indicators}
\end{table}

Thus, the indicator function $\xi(X \cup Y \ni i,j)$ may be written as a sum of indicator functions of mutually exclusive events listed in the table: $\xi(X \cup Y \ni i,j) = \sum_{n=1}^{9}\xi_{n}.$ 
The overall sum that we want to bound can be written as a sum over the nine cases:
\begin{align}
    &\sum_{Z \ni i,j}\norm{h_{Z}} 
    = \sum_{X \cup Y \ni i,j}\norm{[a_{X}, b_{Y}]} \nonumber\\
    &\leq 2\sum_{X}\sum_{Y} \norm{a_{X}}\norm{b_{Y}} \xi(X \cap Y \neq \varnothing)\xi(X \cup Y \ni i,j) \nonumber\\
    &=2\sum_{n=1}^{9}\sum_{X}\sum_{Y} \norm{a_{X}}\norm{b_{Y}} \xi(X \cap Y \neq \varnothing)\xi_{n},
\end{align}
and we will bound each of the nine cases individually. 
We will often eliminate the condition that $X \cap Y \neq \varnothing$, which can only make the sum larger, and introduce an inequality by summing over all sets $X$ or $Y$. To illustrate our technique, consider first the contribution from $\xi_{5}$: 
\begin{align}
    &2\sum_{X}\sum_{Y}\norm{a_X}\norm{b_Y}\xi(X \cap Y \neq \varnothing) \xi_{5} \nonumber\\
    &\leq 2\sum_{X \ni i} \sum_{Y\ni i,j}\norm{a_X}\norm{b_Y} \nonumber\\
    &\leq 2\sum_{X \ni i}\norm{a_X} \frac{b}{\dist(i,j)^{\alpha}} 
     \leq  \frac{2\lambda_0 ab}{\dist(i,j)^\alpha},
\end{align}
where the first inequality comes from ignoring $j\notin X$ and the second comes from $H_{2}$ being a power-law Hamiltonian. 

The bound on the term corresponding to $\xi_7$ follows analogously since we simply switch $i, j$. 
Similarly, the terms corresponding to $\xi_{2},\xi_{3}$ switch only the roles of $X, Y$ compared to $\xi_{5},\xi_{7}$. 
Meanwhile, analyzing the term corresponding to $\xi_1$ yields:
\begin{align}
    &2\sum_{X}\sum_{Y} \norm{a_X}\norm{b_Y}\xi_{1}\xi(X \cap Y \neq \varnothing) \nonumber\\
    &= 2\sum_{X\ni i,j}\sum_{Y \ni i,j} \norm{a_X}\norm{b_Y} \nonumber\\
    &\leq \frac{2ab}{\dist(i,j)^{2\alpha}} \leq \frac{2ab}{\dist(i,j)^{\alpha}},
\end{align}
where we take into account $\dist(i,j)\geq 1$ for all $D$. 

Upper bounding the term corresponding to $\xi_6$ is a bit trickier.
Since $X\cap Y\neq \varnothing$, there exists a site $\ell\neq i,j$ such that $\ell\in X\cap Y$.
Rewriting the term corresponding to $\xi_6$ as a sum over $\ell$, we have:
\begin{align}
    &\sum_{\substack{X\ni i\\X\not\ni j}} \sum_{\substack{Y\ni j\\Y\not\ni i}}\norm{a_X}\norm{b_Y}\xi(X \cap Y \neq \varnothing) \nonumber\\
    &\leq 2\sum_{\ell\neq i, j} \sum_{X\ni i, \ell}\sum_{Y \ni j, \ell} 2\norm{a_X}\norm{b_Y}\nonumber\\
    &\leq 2\sum_{\ell\neq i, j} \frac{a }{\dist(i,\ell)^{\alpha}}\frac{b  }{\dist(\ell,j)^{\alpha}} 
    \leq \frac{2\lambda_1 a b  }{\dist(i,j)^{\alpha}},
\end{align}
where the last inequality comes from the reproducibility condition~\cite{Hastings06}, applicable when $\alpha>D$, and $\lambda_1$ is a constant that depends only on $D,\alpha$.
The term corresponding to $\xi_{8}$ contributes the same as $\xi_{6}$, as it only switches the roles of $i, j$. 

Finally, we bound the terms corresponding to $\xi_4, \xi_9$.
For $\xi_{4}$, we are trying to bound the sum:
    \begin{align}
        \sum_{X \ni i,j} \sum_{Y \not\ni i,j}2\norm{a_{X}}\norm{b_{Y}}\xi(X \cap Y \neq \varnothing).
    \end{align}
The non-empty intersection means that for there to be a nonzero contribution, $\exists \ell \neq i, j$ such that $\ell \in X, Y$. Further note that by assumption the maximum extent of $X$ is $k_{1} + 1$ and therefore there are at most $k_{1} - 1$ sites distinct from $i,j$ where $Y$ can intersect with $X$. 
We bound this as follows:
\begin{align}
            &2\sum_{X \ni i,j} \sum_{Y \not\ni i,j}\norm{a_{X}}\norm{b_{Y}} \xi(X \cap Y \neq \varnothing)\nonumber\\ 
            &\leq 2\sum_{X \ni i, j} \sum_{\substack{\ell\in X\\ \ell\neq i,j}}\sum_{Y\ni \ell}\norm{a_{X}}\norm{b_{Y}}  \nonumber\\
            &\leq 2\sum_{X \ni i, j}\norm{a_{X}} \sum_{\substack{\ell\in X\\\ell\neq i,j}}  \lambda_0 b  
            \leq \frac{2\lambda_0(k_1-1)ab}{\dist(i,j)^\alpha}.
\end{align}
We bound the term corresponding to $\xi_9$ similarly by switching the role of $X,Y$. 
Collecting everything, we have the lemma with $\lambda = 2(6\lambda_0+2\lambda_1+1)$.
\end{proof}\subsection{Bounds on discrete sums}\label{sec:sums}
In this section, we provide bounds on some discrete sums used in the main text.
\begin{lemma}\label{LEM_Sum_inv_mult}
For all $1\leq k\leq q$, we have the following inequalities:
\begin{align}
	&\sum_{\substack{
			1\leq i_1,\dots,i_k\leq q\\
			i_1+\dots+i_k=q}} \prod_{j=1}^k i_j! \leq \frac{q!}{(k-1)!},\label{LEM_sum_i_part1}\\
	&\sum_{\substack{
			0\leq i_1,\dots,i_k\leq q\\
			i_1+\dots+i_k=q}} \prod_{j=1}^k i_j! \leq 2^k q!.\label{LEM_sum_i_part2}
\end{align}
\begin{proof}
We first bound
    \begin{align}
        \sum_{\substack{
			1\leq i_1,\dots,i_k\leq q\\
			i_1+\dots+i_k=q}} \prod_{j=1}^k i_j! \leq \binom{q-1}{k-1} \max_{\substack{
			1\leq i_1,\dots,i_k\leq q\\
			i_1+\dots+i_k=q}} \prod_{j=1}^k i_j!. 
    \end{align}
For positive integers $a \geq b$, we have $(a+b-1)! = a!(a+b-1)\cdots(a+1) \geq a! b!$ with equality if either $a, b = 1$. This implies that the maximal product occurs for some $i_j = q -k+1$ and $i_{k\neq j} =1$ (we omit the simple proof by induction), yielding
    \begin{align}
        &\binom{q-1}{k-1} \max_{\substack{
			1\leq i_1,\dots,i_k\leq q\\
			i_1+\dots+i_k=q}} \prod_{j=1}^k i_j! \nonumber\\
            &\leq  \frac{(q-1)!}{(k-1)!(q-k)!} (q-k+1)! \nonumber\\
            &\leq \frac{(q-1)!}{(k-1)!}(q-k+1) \leq \frac{q!}{(k-1)!},
    \end{align}
as $k \neq 0$ by the summation restrictions. \cref{LEM_sum_i_part2} is essentially the same as \cref{LEM_sum_i_part1} with some indices allowed to be 0. For example, if $i_1 = 0$ while the other $i$ are nonzero, it is just \cref{LEM_sum_i_part1} with $k \rightarrow k-1$.
This part of the sum is then crudely upper bounded by $q!$, while summing over all possible choices of zero indices leads to a factor $2^k$.
\end{proof} 
\end{lemma}

\begin{corollary}
For all $1\leq k\leq q$, we have:
\begin{align}
        \sum_{\substack{
            1\leq i_1,\dots,i_k\leq q_0\\
            i_1+\dots+i_k=q_0}} \prod_{j=1}^k (i_j-1)! \leq 2^{k}(q_{0}-k)!.
\end{align}
\begin{proof}
Define $p_{j} = i_{j} - 1$ such that $0 \leq p_{j} \leq q_{0}-1$ and  $p_{1} + \cdots + p_{k} = q_{0} - k$. This second condition implies that we may simplify the first condition to $0 \leq p_{j} \leq q_{0}-k$. Therefore: 
\begin{align}
        &\sum_{\substack{
            1\leq i_1,\dots,i_k\leq q_0\\
            i_1+\dots+i_k=q_0}} \prod_{j=1}^k (i_j-1)!     \nonumber\\
            =&
        \sum_{\substack{
            0\leq p_1,\dots,p_k\leq q_0-k\\
            p_1+\dots+p_k=q_0-k}} \prod_{j=1}^k p_{j}! \leq 2^{k}(q_{0}-k)!, 
\end{align}
where the last inequality is from \cref{LEM_sum_i_part2}. 
\end{proof}
\end{corollary}

\begin{lemma}\label{lem:c1}
For all $1\leq k\leq q$, we have:
\begin{equation}
\sum_{k = 1}^{q_{0}} \frac{2^{k}q_{0}^{k}c^{-k}(q_{0}-k)!}{q_{0}!k!} \leq \frac{e}{\sqrt{2\pi}}(e^{2e/c}-1)
\end{equation}
\begin{proof}
Using Stirling's approximation, $\sqrt{2\pi} n^{n + \frac{1}{2}} e^{-n}\leq n! \leq e n^{n + \frac{1}{2}}e^{-n}$ for $q_0!$ and $(q_0-k)!$, we can bound:
    \begin{align}
        &\sum_{k = 1}^{q_{0}} \frac{2^{k}q_{0}^{k}c^{-k}}{k!}\frac{(q_{0}-k)!}{q_{0}!} \nonumber\\
        &\leq \sum_{k = 1}^{q_{0}} 
            \frac{2^{k}q_{0}^{k}c^{-k}}{k!}
            \frac{e}{\sqrt{2\pi}}
            \frac{(q_0-k)^{q_0-k}}{q_{0}^{q_{0}}}
            \frac{\sqrt{q_0-k}}{\sqrt{q_0}}
            \frac{e^{-(q_0-k)}}{e^{-q_{0}}} \nonumber\\
        &\leq 
            \frac{e}{\sqrt{2\pi}} \sum_{k = 1}^{q_{0}} 
            \frac{2^{k}e^{k}c^{-k}}{k!}
            \underbrace{\frac{(q_0-k)^{q_0-k}}{q_{0}^{q_{0}-k}}}_{\leq 1}
            \underbrace{\frac{\sqrt{q_0-k}}{\sqrt{q_0}}}_{\leq 1}\nonumber\\    
        &\leq 
            \frac{e}{\sqrt{2\pi}} \sum_{k = 1}^{\infty} 
            \frac{2^{k}e^{k}c^{-k}}{k!}        
        = \frac{e}{\sqrt{2\pi}} (e^{2e/c}-1) .
    \end{align}
We note that the bound approaches 0 as $c\rightarrow\infty$.
\end{proof}
\end{lemma}

\begin{lemma} \label{lem:forElse}
For $D\in \mathbb N_{>0}, r_*>1,0<\eta<1$
\begin{align}
        \sum_{r> r_*}   r^{D-1}e^{-r^\eta}
        \leq \frac{2}{\eta}2^{D/\eta}\Gamma(D/\eta) r_*^{D} e^{-r_*^{\eta}},
\end{align}
where $\Gamma$ is the Gamma function.
\end{lemma}
\begin{proof}
Let $f(r) = r^{D-1}e^{-r^\eta}$.
Our strategy is to upper bound $\sum_{r>r_*}f(r)$ by an integral. 
For $r\in (0,\infty)$, $f$ has a maximum at $r=r_0 =(D-1)^{1/\eta} \eta^{-1/\eta} $.
Let $r_{0}^- = \floor{r_0}$ and $r_{0}^+ = r_0^-+1>r_0$.
Then, the function $f(r)$ is increasing for $r\in (r_*,r_0^-)$ and decreasing for $r\geq r_0^+$.
Therefore, we can upper bound:
\begin{align}
    \sum_{r>r_*} f(r) &\leq \int_{r_*}^{r_0^-} f(r) dr + \int_{r_0^+}^{\infty} f(r) dr + f(r_0^-)+f(r_0^+)\nonumber\\
    &\leq \int_{r_*}^{r_0^-} f(r) dr + \int_{r_0^+}^{\infty} f(r) dr + 2\int_{r_0^-}^{r_0^+} f(r)dr\nonumber\\
    &\leq 2\int_{r_*}^{\infty} f(r) dr,
\end{align}
where we use the fact that $f(r)$ is concave between $r_0^-$ and $r_0^+$ to bound the first line by the second line.
Next, to bound the integral, we make a change of variable to $x = r^{\eta}$ so that
\begin{align}
    2\int_{r_*}^{\infty} f(r) dr  
    &= 2\int_{r_*}^{\infty} r^{D-1}e^{-r^{\eta}} dr\nonumber\\
    &=\frac{2}{\eta}\int_{x_*}^{\infty} x^{\frac{D-\eta}{\eta}} e^{-x}  dx\nonumber\\
    &\leq\frac{2}{\eta}\int_{x_*}^{\infty} x^{\beta} e^{-x}  dx
    \leq \frac{2}{\eta}2^\beta \beta! x_*^{\beta}e^{-x_*} \nonumber\\
    &= \frac{2}{\eta}2^\beta\beta! r_*^{\eta \beta} e^{-r_*^{\eta}} \nonumber\\
    &\leq \frac{2}{\eta}2^{D/\eta}\Gamma(D/\eta) r_*^{D} e^{-r_*^{\eta}},
\end{align}
where $x_{*} = r_{*}^{\eta}$, $\beta = \ceil{{(D-\eta)}/{\eta}}\leq D/\eta$ is an integer, and $\Gamma$ is the Gamma function.
Note that we have also used a bound for the integral 
\begin{align}
    \int_{x_*}^{\infty} x^{\beta} e^{-x}  dx\leq 2^\beta \beta! x_*^\beta e^{-x_*},
\end{align}
which can be proven inductively on $\beta$ for all $\beta\geq 0$ and $x_*\geq 2$.
Indeed, the inequality is trivial for $\beta = 0$.
Suppose the inequality holds for $\beta-1$, using integration by parts, we have
\begin{align}
    \int_{x_*}^{\infty} x^{\beta} e^{-x}  dx 
    &= x_*^\beta e^{-x_*} + \beta \int_{x_*}^{\infty} x^{\beta -1} e^{-x} dx\nonumber\\
    &\leq x_*^\beta e^{-x_*} + \beta 2^{\beta-1}(\beta-1)! x_*^{\beta-1} e^{-x_*}\nonumber\\
    &\leq 2^{\beta-1}\left(\frac{1}{2^{\beta-1}\beta!}+\frac{1}{x_*}\right)\beta!x_*^{\beta} e^{-x_*} \nonumber\\
    &\leq 2^\beta\beta!x_*^{\beta} e^{-x_*}, 
\end{align}
where the terms inside the bracket in the second to last line is always less than or equal to 2 for all $x_*\geq 1$ (corresponding to $r_{*} > 1$).
\end{proof}

\end{document}